\documentclass[12pt,a4paper]{article}

\usepackage[english]{babel}
\usepackage[top=2cm,bottom=2cm,left=2cm,right=2cm]{geometry}

\usepackage{amsmath}
\usepackage{amsfonts}
\usepackage{graphicx}
\usepackage[colorlinks=true, allcolors=blue]{hyperref}
\usepackage{enumerate}  
\usepackage{float}
\usepackage{sidecap}
\usepackage{subcaption}
\usepackage{booktabs}
\usepackage{authblk}

\usepackage{amsthm}
\usepackage{hyperref}
\usepackage[capitalise]{cleveref}
\usepackage{autonum}

\usepackage{algorithm}
\usepackage{algpseudocode}

\newtheorem{theorem}{Theorem}[section]
\newtheorem{definition}[theorem]{Definition}
\newtheorem{lemma}[theorem]{Lemma}

\newtheorem{corollary}[theorem]{Corollary}

\newcommand{\msf}[1]{\mathsf{#1}}
\newcommand{\bbR}{\mathbb{R}}
\newcommand{\bbC}{\mathbb{C}}
\newcommand{\bbE}{\mathbb{E}}
\newcommand{\bbN}{\mathbb{N}}

\newcommand{\cH}{\mathcal{H}}
\newcommand{\cE}{\mathcal{E}}
\newcommand{\cN}{\mathcal{N}}
\newcommand{\cB}{\mathcal{B}}
\newcommand{\cS}{\mathcal{S}}
\DeclareMathOperator{\tr}{tr}
\DeclareMathOperator{\re}{Re}
\DeclareMathOperator{\im}{Im}
\DeclareMathOperator{\diag}{diag}

\DeclareMathOperator{\desc}{\Delta}
\DeclareMathOperator{\Desc}{\mathsf{Desc}}
\DeclareMathOperator{\cov}{\Gamma}
\DeclareMathOperator{\var}{Var}

\newcommand{\proj}[1]{|#1\rangle\langle#1|}

\usepackage{algpseudocode}
\usepackage{mdframed}

\newcommand{\Aoverlaptriple}{\mathsf{overlaptriple}}
\newcommand{\Aoverlap}{\mathsf{overlap}}

\newcommand{\Asqueezing}{\mathsf{squeezing}}
\newcommand{\Aunitary}{\mathsf{applyunitary}}
\newcommand{\Ameasure}{\mathsf{postmeasure}}
\newcommand{\Ameasureprob}{\mathsf{measureprob}}
\newcommand{\Aprob}{\mathsf{prob}}

\newcommand{\Ameasureprobexact}{\mathsf{measureprobexact}}

\newcommand{\Ameasureprobapproximate}{\mathsf{measureprobapproximate}}
\newcommand{\Aexactnorm}{\mathsf{exactnorm}}
\newcommand{\Afastnorm}{\mathsf{fastnorm}}
\renewcommand*{\Re}{\mathsf{Re}}
\newcommand*{\ket}[1]{|#1\rangle}
\newcommand*{\bra}[1]{\langle #1|}

\newcommand*{\Sgauss}{\mathcal{S}^{\mathrm{Gauss}_n}}
\newcommand*{\cT}{\mathcal{T}}
\newcommand*{\Var}{\mathsf{Var}}

\begin{document}

\title{Classical simulation of non-Gaussian bosonic circuits}
\author{Beatriz Dias}
\author{Robert K{\"o}nig}
\affil[1]{Department of Mathematics, School of Computation, Information and Technology, Technical University of Munich, Garching}
\affil[2]{Munich Center for Quantum Science and Technology, Munich, Germany}
\date{\today}

\maketitle

\begin{abstract}
    We propose efficient classical algorithms which (strongly) simulate the action of bosonic linear optics circuits applied to superpositions of Gaussian states. Our approach relies on an augmented covariance matrix formalism to keep track of relative phases between individual  terms in a linear combination. This yields an exact simulation algorithm whose runtime is polynomial in the number of modes and the size of the circuit, and quadratic in the number of terms in the superposition. We also present a faster approximate randomized algorithm whose runtime is linear in this number.  Our main building blocks are a formula for the triple overlap of three Gaussian states and a fast algorithm for estimating the norm of a superposition of Gaussian states up to a multiplicative error.  Our construction borrows from earlier work on simulating quantum circuits in finite-dimensional settings, including, in particular, fermionic linear optics with non-Gaussian initial states and Clifford computations with non-stabilizer initial states. It provides algorithmic access to a practically relevant family of non-Gaussian bosonic circuits.
\end{abstract}

\tableofcontents

\section{Introduction}

If one is to believe in Feynman's vision~\cite{feynmanSimulatingPhysicsComputers1982}, understanding quantum dynamics may one day be a prime field of application of a working universal quantum computer. Without access to such a powerful computational tool, however, one has to resort to efficient, i.e., polynomial time, classical methods  that can be used to
study the behavior of certain quantum many-body systems evolving  e.g., under unitaries and measurement. 
A difficulty here is to identify  classes of quantum computations which allow for efficient classical simulation algorithms. Well-known examples include matchgate circuits~\cite{doi:10.1137/S0097539700377025,Jozsa_2008}, Gaussian dynamics of either fermions~\cite{PhysRevA.65.032325,knill2001fermionic} or bosons~\cite{WANG_2007}, and Clifford circuits (stabilizer computations)~\cite{PhysRevA.70.052328}. For each of these classes of quantum dynamics, polynomial-time simulation methods have been constructed.

For bosonic systems consisting of say~$n$ modes, low-order  moments of the  mode operators of a state (that is, expectation values of constant-degree polynomials of the quadratures) can easily be kept track of under Gaussian dynamics (in particular, Gaussian unitaries and measurement): This requires only a polynomial-time computation, and implies that (bosonic) linear optics, i.e., computations with such operations starting from an initial Gaussian state can be simulated efficiently. This is the case for both strong simulation (where the goal is to compute the output probability for a given measurement outcome) as well as weak simulation (where the goal is to sample from the output distribution defined by measurement), assuming the classical algorithm has access to independent unit-variance centered Gaussian random variables.

While the framework of quantum linear optics is amenable to efficient classical simulation algorithms, it is not sufficient to describe a number of information-processing applications of interest. Consider for example quantum error correction. To quote from the abstract of a pioneering paper~\cite{NisetFiurasekCerf09}, {\em 
``[...] Gaussian operations are of no use for protecting Gaussian states against Gaussian errors [...]''} (see~\cite{GKPToric} for another no-go result). 
It is thus necessary to  augment quantum optics by non-Gaussian resource states.
For instance, displacement errors can be protected against by encoding a qubit into an (approximate) Gottesman-Kitaev-Preskill (GKP) code~\cite{gkp}. 
Other proposals for protecting against displacement errors include the family of oscillator-to-oscillator codes~\cite{nohoscillatortooscillator}, where encoding proceeds by applying a Gaussian unitary to a bosonic mode, with a  collection of auxiliary systems prepared in GKP-states. To protect against photon loss (another ubiquitous example of Gaussian errors), coherent 
cat codes~\cite{PhysRevA.94.042332, PhysRevLett.119.030502,PhysRevLett.111.120501,allopticalcatcode} have been developed.

What do GKP-, oscillator-to-oscillator- and coherent cat-codes have in common? Most importantly, they are experimentally feasible, and corresponding schemes have been realized in the lab. What makes this possible is the simple structure of the associated circuits: While the initial states are non-Gaussian, subsequent operations (e.g., for applying logical gates) typically only involve Gaussian unitaries.  Second, analyzing their properties presents a difficult challenge for theory. In particular, unlike for Gaussian initial states, it is unclear in general how to construct classical simulation algorithms which work for non-Gaussian initial states.  It is clear that the hardness, i.e., complexity of classical simulation is a function of these states. 
Our goal is to quantify this, i.e., we ask: What is the computational effort required to (strongly) simulate a Gaussian quantum circuit on a non-Gaussian input state? 

\subsection{Prior work \label{sec:priorwork}}
In an effort to go beyond this framework of linear optics, several alternative methods for simulating bosonic dynamics have been developed. For example, it has been shown that if 
a circuit of Gaussian unitaries is applied to an initial state with a positive Wigner function, then it is possible to efficiently sample from the output distribution of a subsequent measurement~\cite{Veitch_2013,PhysRevLett.109.230503}. Since Gaussian states have a positive Wigner function, this constitutes a strict generalization of the covariance matrix formalism to a larger class of states. One way of putting this result is to say that Wigner function negativity is a precondition for a quantum advantage~\cite{Veitch_2012}.
Along similar lines, Ref.~\cite{PRXQuantum.2.040315} considered non-Gaussian states whose Wigner function can be written as a linear combination of Gaussian Wigner functions. Several natural non-Gaussian states were shown to be describable in this way. When using complex Gaussians, this framework can also describe superpositions of pure Gaussian states. Update rules for such descriptions under Gaussian operations were established.

In more recent work~\cite{PhysRevResearch.3.033018,PhysRevLett.130.090602}, a different set of  non-Gaussian initial states for subsequent Gaussian operations have been shown to also lead to a class of computations that can be simulated efficiently in the strong sense. The authors of Refs.~\cite{PhysRevResearch.3.033018,PhysRevLett.130.090602} give algorithms which work for initial pure states whose coefficients in the Fock basis have bounded support. Their simulation algorithms have polynomial runtime in the size of this support. Additionally, the algorithms proposed in Refs.~\cite{PhysRevResearch.3.033018} and~\cite{PhysRevLett.130.090602} have exponential runtime respectively in the mean photon number of the initial state and in a non-Gaussianity measure called the stellar rank. 

Another instance of an efficient classically simulable computation put forward in Refs.~\cite{Calcluth2022efficientsimulation,PhysRevA.107.062414} includes infinitely squeezed GKP-states together with a certain set of Gaussian unitary operations and homodyne measurements. In subsequent work \cite{PhysRevA.107.062414,calcluth2024sufficient}, the authors establish that certain resourceful/magic states (such as the vacuum state) can promote such a computation to universality.

For GKP- and cat-codes (and more generally when dealing with infinite-dimensional Hilbert spaces), perhaps the most widely used approach to numerical simulation (see e.g.,~\cite{performancestructurebosonic}) is to project onto a finite-dimensional Hilbert space (and subsequently do, e.g., state vector simulation). Concretely, e.g., in the work~\cite{performancestructurebosonic}, this 
is done by replacing GKP- and cat-states
by finite Fock state superpositions~$\sum_{n=0}^{N_{\max }} c_n|n\rangle$ (with~$N_{\max }$ as large as manageable). While this is sufficient for the analysis of single-mode bosonic codes, it is desirable to find better approaches that may be more scalable. 
In this context, the consideration of number (Fock) states also does not appear to be ideal as these have non-trivial dynamics under Gaussian operations.
An alternative approach proposed in Ref.~\cite{Marshall:23} considers Gaussian unitary evolution and Fock basis measurements applied to non-Gaussian states decomposed into a superposition of coherent states, defining a notion called coherent state rank. These algorithms are shown to be more efficient than algorithms considering Fock state superpositions.

\subsection{Towards
new simulation algorithms}

Gaussian states and operations define a well-known class of many-body quantum dynamics that is amenable both to analytical considerations, as well as efficient classical simulation.  Motivated by this, we ask: can the set of initial states to a Gaussian quantum circuit be extended in a different direction beyond Gaussian states, namely superpositions of Gaussian states, while maintaining efficient (strong) simulability? 

In more detail, consider the following type of~$n$-mode bosonic quantum computation, which augments linear optics by the use of an initial state~$\Psi$ form a certain set of states~$\mathcal{S}\subset L^2(\mathbb{R})^{\otimes n}$.
Here and throughout we use the notation~$[\ell] = \{1,\ldots, \ell\}$ for~$\ell\in\bbN$.
\begin{enumerate}
\item
Let~$\Psi\in\mathcal{S}$ be an arbitrary initial state.
\item\label{it:stepunitarygaussian}
For each~$t\in \{1,\ldots,s\}$, apply a unitary~$U^{(t)}$ from one of the following Gaussian unitaries
\begin{enumerate}[(i)]
\item
an~$n$-mode displacement operator~$D(\alpha)$, $\alpha\in\mathbb{C}^n$,
\item
a single-mode phase shifter~$F_j(\phi)$, $\phi\in\mathbb{R}$
applied to the~$j$-th mode,
\item
a single-mode squeezing operator~$S_j(z)$, $z\geq 0$ applied to the~$j$-th mode, for~$j\in [n]$, or
\item
a beamsplitter~$B_{j,k}(\omega)$, $\omega\in\mathbb{R}$, $j\neq k$
applied to the~$j$- and~$k$-th modes. 
\end{enumerate}
We refer to~\cref{sec:gaussian-unitary-op} for a definition of these operators.
Here the mode~$j=j^{(t)}$ respectively the pair of modes~$(j,k)=(j^{(t)},k^{(t)})$ depends on the time step~$t$, as do the parameters~$\alpha=\alpha^{(t)}$, $\phi=\phi^{(t)}$ etc. For notational clarity, we often suppress this dependence.
\item\label{it:stepmeasurementgaussian}
Measurement of~$k\leq n$ modes of the final state~$U^{(s)}\cdots U^{(1)}\Psi$ as follows: Each mode~$k\in [n]$ is measured using a heterodyne measurement yielding a result~$\alpha_j\in\mathbb{C}$.
\end{enumerate}
The consideration of the first~$k$ modes for the measurement is  for convenience only and presents no loss of generality. Similarly, the restriction to computations where measurements are only applied at the end of the computation is for simplicity. Our simulation algorithms also deal with post-measurement states and can thus be applied to computations  which involve mid-circuit measurements, or to compute probabilities of marginal distributions when measuring several modes.

In summary, this quantum algorithm produces a sample~$\alpha=(\alpha_1,\ldots,\alpha_k)\in\mathbb{C}^k$,  from a distribution with density
\begin{align}
p(\alpha)&=\frac{1}{\pi^k}\left|\left(\bra{\alpha}\otimes I^{\otimes n-k}\right), U^{(s)}\cdots U^{(1)}\Psi\rangle\right|^2\ ,
\end{align}
where we write~$\ket{\alpha}=\ket{\alpha_1}\otimes\cdots \otimes \ket{\alpha_k}$ for a product of coherent states.   Since the set of unitaries considered in 
Step~\ref{it:stepunitarygaussian} generates all Gaussian unitaries, this class of computations encompasses the mentioned (known) examples~-- the set of Gaussian computations (with initial state~$\Psi$ belonging to the set~$\mathcal{S}=\mathrm{Gauss}_n$ of pure Gaussian states), as well as the extension by initial states with positive Wigner functions, the set of computations starting from states 
with  finite support in the Fock-basis and bounded photon number  (with~$\mathcal{S}$  defined  correspondingly), and the set of computations starting from superpositions of coherent states~-- as well as the direction taken here, i.e., computations starting from superpositions of Gaussian states.

\subsection{Our contribution}

Here we are interested in the set

\begin{align}
\Sgauss(\chi):=\left\{\Psi\in L^2(\mathbb{R})^{\otimes n}\ \Bigg| \
\|\Psi\|=1, \Psi=\sum_{j=1}^\chi c_j\psi_j\ ,\psi_j\in \mathrm{Gauss}_n
\right\}
\end{align}
which are linear combinations of~$\chi$ pure Gaussian states.  Our first main result is the following:
\newcommand*{\Aexactsimulation}{\mathsf{SimulateExactly}}
\newcommand*{\Aapproximatesimulation}{\mathsf{SimulateApproximately}}
\begin{theorem}[Exact strong simulation]
There is a classical algorithm~$\Aexactsimulation$ which, given
\begin{enumerate}[(i)]
\item
a classical description of a state~$\Psi\in\Sgauss$, 
\item
classical descriptions of~$s$ Gaussian unitaries~$\{U^{(t)}\}_{t=1}^s$,
\item
a vector~$\alpha\in\mathbb{C}^k$, 
\end{enumerate}
computes the value~$p(\alpha)$ exactly in time~$O(s\chi^2n^3)$.
\end{theorem}
The constant in the $O(\cdot)$-notation depends on the Gaussian unitaries~$\{U^{(t)}\}_{t=1}^s$ applied, see Section~\ref{sec:algorithms-gaussian}.
Here the runtime of the algorithm~$\Aexactsimulation$ is  defined as the number of elementary arithmetic operations over the reals, and we assume that we have arbitrary-precision arithmetic available.  This precision is needed since we are not making any assumptions about the initial state~$\Psi$ (except for the fact that it is a linear combination of Gaussian states with~$\chi$ terms): The algorithm~
$\Aexactsimulation$ involves real numbers that are essentially the inverse of the energy (mean photon number) of~$\Psi$. In particular, if no energy bound is imposed, these scalars can become arbitrarily small. 

Imposing a mean number constraint on~$\Psi$ is natural from a physical perspective, and also  circumvents the need for infinite-precision arithmetic. Another more important aspect is that it allows us to improve the runtime to a linear instead of quadratic scaling in~$\chi$. We give a corresponding algorithm~$\Aapproximatesimulation$. In contrast to~$\Aexactsimulation$, it is a probabilistic algorithm and provides only an approximation to the value~$p(\alpha)$ of interest. 
 
To state its properties, it is convenient to 
introduce the set of states 
\begin{align}
\Sgauss_N(\chi)&:=\left\{\Psi\in \Sgauss(\chi) \ \big|\  \langle \Psi,\sum_{j=1}^\chi a_j^\dagger a_j\Psi\rangle\leq N\right\}
\end{align}
which have mean photon number bounded by a constant~$N$ and are  linear combinations  of~$\chi$ pure Gaussian states. For
 a sequence~$\{U^{(t)}\}_{t=1}^s$ of Gaussian unitaries, we also need a measure of the total amount of squeezing introduced, namely
\begin{align}
z_{\mathsf{tot}}&=\sum_{t=1}^s z^{(t)}\ .
\end{align}
Here we set~$z^{(t)}=1$ if the unitary~$U^{(t)}$ at time step~$t$ is not a single-mode squeezing operator. We note that this quantity provides the upper bound~$Ne^{2z_{\mathsf{tot}}}$ on the mean photon number of the final state~$\ket{\Psi_{\textrm{final}}}:=U^{(s)}\cdots U^{(1)}\ket{\Psi}$, assuming that this state is centered, see~\cref{sec:fastnormestimationalg}. We then have the following:
\begin{theorem}[Approximate strong simulation]
Let~$p_f\in (0,1)$ and~$\varepsilon>0$ be arbitrary. Then there is a classical probabilistic algorithm~$\Aapproximatesimulation$ which, given
\begin{enumerate}[(i)]
\item
a classical description of a state~$\Psi\in\Sgauss_N(\chi)$, 
\item
classical descriptions of~$s$ Gaussian unitaries~$\{U^{(t)}\}_{t=1}^s$,
\item
a typical vector~$\alpha\in\mathbb{C}^k$ of measurement results,
\end{enumerate}
computes the value~$p(\alpha)$ to within a multiplicative error~$\varepsilon$ 
with probability at least~$1-p_f$. It has runtime bounded by~$O(s\chi n^3 Ne^{2z_{\mathsf{tot}}}/(\varepsilon^3 p_f))$.
\end{theorem}
In the algorithm~$\Aapproximatesimulation$, the necessary arithmetic precision is set by~$\varepsilon$ and the inverse of the mean photon number of the final state~$\ket{\Psi_{\textrm{final}}}$.
We refer to~\cref{sec:fastnormestimationalg}
for a thorough discussion of the notion of typicality for a measurement result~$\alpha\in\mathbb{C}^k$ (see~\cref{lem:typicalsubsetlemma}).  It amounts to the statement that~$\alpha$~belongs to a set that has high probability mass under~$p_\Psi$, namely a ball of large radius compared to the energy of~$\ket{\Psi_{\textrm{final}}}$, and the normalized post-measurement state for the measurement  outcome~$\alpha$ has a moderate mean photon number (again compared to the state~$\ket{\Psi_{\textrm{final}}}$). 

The computational complexity of our algorithms is tied to the (bosonic) Gaussian rank of the initial state~$\Psi$,  i.e., the minimal integer~$\chi\in\mathbb{N}$ such that~$\Psi\in \Sgauss(\chi)$. This definition is analogous to  the stabilizer rank \cite{bravyiImprovedClassicalSimulation2016,bravyiImprovedClassicalSimulation2016}, as well as to the fermionic Gaussian rank \cite{dias2023classical,cudby2023gaussian}. Assuming that the initial state is provided as a superposition of Gaussian states with the minimum number~$\chi$ of terms, the algorithms~$\Aexactsimulation$ and~$\Aapproximatesimulation$ scale quadratically and linearly in~$\chi$, respectively. Starting from this observation, one may introduce the notion of the bosonic Gaussian extent of~$\Psi$, similar to the way the stabilizer~\cite{bravyiSimulationQuantumCircuits2019a} and fermionic Gaussian extent~\cite{dias2023classical,cudby2023gaussian,reardon-smithImprovedSimulationQuantum2023} are defined. This quantity is typically much smaller than the Gaussian rank and is a proxy for an approximate notion of Gaussian rank. The runtime of an appropriately modified version of the algorithm~$\Aapproximatesimulation$   scales linearly 
in this quantity.

Classical strong simulation is indeed stronger than weak simulation in the sense that an algorithm for strong simulation gives an algorithm for weak simulation, provided that there is an efficient algorithm to sample from the output probability distribution associated to measurement \cite{Pashayan2020fromestimationof}. This statement carries over to continuous variable computations by discretizing the output probability density, provided there is an efficient binning of the space of outcomes and the discretized probabilities can be computed efficiently from the output probability density and sampled from efficiently \cite{PhysRevResearch.3.033018}.
Then, if provided, an efficient binning of the outcome space as well as efficient algorithms to compute and sample discretized probabilities from the output probability density augment our classical simulation methods in a way that allows weak simulation.

Our work mirrors earlier findings in the context of stabilizer \cite{bravyiImprovedClassicalSimulation2016,bravyiSimulationQuantumCircuits2019a,pashayanFastEstimationOutcome2022b}, matchgate and fermionic linear optics computations \cite{dias2023classical,reardon-smithImprovedSimulationQuantum2023}. It makes an experimentally important class of non-Gaussian circuits accessible to numerical simulation.

\subsection{Application of our algorithm
to bosonic error correction}
We note that there is by now a wide variety of classical simulation algorithms for bosonic dynamics (see the discussion in Section~\ref{sec:priorwork}). Since they typically apply to different sets of (initial) states and operations, a direct comparison of their merits is non-trivial.  Rather, the algorithm best suited for classical simulation heavily depends on the application considered.

As already mentioned, we believe that an area of application particularly suited for our algorithm is that of bosonic error correction. Concretely, 
an (approximate) Gottesman-Kitaev-Preskill (GKP) code~\cite{gkp} encodes a qubit into the two-dimensional subspace of an oscillator defined by the two states
\begin{align}
\ket{\overline{b}}_{\mathrm{GKP}}&\propto \sum_{s=-\infty}^{\infty} e^{-\frac{\kappa^2}{2} \left((2 s+1) \sqrt{2\pi}\right)^2} e^{-i(2 s+b)P \sqrt{2\pi}}\left|\psi_\Delta \right\rangle\ ,\label{eq:gkpcodestate}
\end{align}
for $b\in \{0,1\}$, 
where~$\psi_\Delta(x)=e^{-x^2/(2\Delta^2)}/(\pi\Delta^2)^{1/4}$ is a Gaussian wavefunction, and where $(\Delta,\kappa)$ describes the shape (squeezing)  of the GKP-state, see Fig.~\ref{fig:gkpstate}. This code protects against (small) random displacement errors in phase space. Fault-tolerant quantum computation with a CV system can be realized by means of copies of these GKP-states and Gaussian operations: Indeed, GKP-states and such operations provide computational universality~\cite{PhysRevLett.123.200502}.

Because of the Gaussian envelope
in Eq.~\eqref{eq:gkpcodestate} (see Fig.~\ref{fig:gkpstate}), GKP-states can be approximated well by a finite linear combination of Gaussian states. The behavior of a fault-tolerant quantum circuit 
with an initial state that is a tensor product of states of the form~\eqref{eq:gkpcodestate} can therefore be studied using our algorithm.

Similar considerations apply to so-called
coherent cat-states, which have been proposed for protecting against photon loss (another ubiquitous example of Gaussian errors), see e.g.,~\cite{PhysRevA.94.042332, PhysRevLett.119.030502,PhysRevLett.111.120501,allopticalcatcode}. Here a qubit is encoded into the span of the two states
\begin{align}
\ket{\overline{b}}_{\mathrm{CAT}}&\propto \ket{i^b\alpha}+\ket{i^b(-\alpha)}\ ,\label{eq:coherentcat}
\end{align}
where $b\in \{0,1\}$. Since the states~\eqref{eq:coherentcat} only involve two Gaussian states in the superposition (i.e., have Gaussian rank~$2$), they are even more compatible with our simulation method.

%(These  are approximately orthogonal for sufficiently large~$|\alpha|, \alpha\in\mathbb{C}$.) 

Another related example of a circuit simulable with our approach 
is the encoding circuit of an oscillator-to-oscillator code~\cite{nohoscillatortooscillator}: Here the encoding proceeds by applying a Gaussian unitary to a bosonic mode, with a  collection of auxiliary systems prepared in GKP-states, see Fig.~\ref{fig:oscosccodecircuit}. Our circuit applies whenever the state to be encoded is itself decomposed into a finite linear combination of Gaussian states.  We note that such corresponding states and circuit have been implemented experimentally, see e.g.,~\cite{campagne-ibarcqQuantumErrorCorrection2020}.

\begin{figure}[!b]
  \centering
  \includegraphics[width=0.4\linewidth]{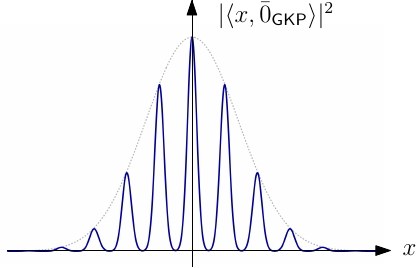}
  \caption{\justifying The amplitude~$|\langle x,\overline{0}_{\mathrm{GKP}}\rangle|^2$ of the (approximate) GKP state~$\ket{\overline{0}}_{\mathrm{GKP}}$ in the position-basis~$\{\ket{x}\}_{x\in\mathbb{R}}$\label{fig:gkpstate}.}
  \label{fig:gkpstate}
\end{figure}

\begin{figure}[!b]
    \centering
  \includegraphics[width=0.45\linewidth]{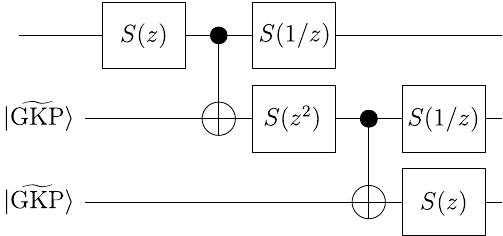}
  \caption{\justifying The $3$-mode GKP-squeezed repetition encoder
by Noh et al.~\cite{nohoscillatortooscillator}\label{fig:oscosccodecircuit}}
\label{fig:gaussiansuperpositioncircuit}
\end{figure}
%Examples of Gaussian-superposition linear optics circuits

\subsection{Technical building blocks}

Our work introduces new methods for describing and working with superpositions of Gaussian states that may be of independent interest. 
One of our main technical tools is a formula for
the expression
\begin{align}
    \tr\left( D(\alpha) \rho(\cov_1, d_1) \rho(\cov_2, d_2) \rho(\cov_3,d_3) \right)
\end{align}
where~$\alpha\in\mathbb{C}^{n}$ specifies a displacement operator~$D(\alpha)$, and where~$\rho(\Gamma,d)$ denotes a Gaussian state with covariance matrix~$\Gamma$ and displacement vector~$d$. Using this formula, which is a generalization of an expression previously derived in~\cite{PhysRevA.61.022306} for the case of centered Gaussian states and~$\alpha$ set to zero, one can compute expressions of the form
\begin{align}
\label{eq:overlaptripleintro}
\langle \psi_3,D(\alpha)\psi_1\rangle\langle\psi_1,\psi_2\rangle\langle\psi_2,\psi_3\rangle
\end{align}
for three Gaussian states~$\psi_j$, $j=1,2,3$, in terms of their covariance matrices and displacement vectors, see~\cref{prop:trace3statesWWeylPsi}. 
Our construction then relies on using a description of Gaussian states (including their global phase) that can be combined with this formula to evaluate inner products (including phases). We have previously used such an approach in the context of non-Gaussian operations using fermions~\cite{dias2023classical}.
The use of expressions of the form~\eqref{eq:overlaptripleintro} allows us to keep track of the relative phases between Gaussian states in a superposition. We note that there may be other ways of achieving this. For example, Ref.~\cite{yao2023design} provides formulas for the global phase of a composition of Gaussian unitaries, which could serve this purpose. Such an approach has been pursued in the fermionic context by Ref.~\cite{reardon-smithImprovedSimulationQuantum2023}.

To describe a single pure Gaussian state~$\ket{\psi}\in\mathrm{Gauss}_n$ including its phase, we augment the standard description of the corresponding density operator~$\proj{\psi}$ (by its first and second moments, i.e., the covariance matrix and displacement vector)  by a single scalar~$r\in\mathbb{C}$: This is the overlap~$r(\psi)=\langle \alpha(\psi),\psi\rangle$ with a coherent state~$\ket{\alpha(\psi)}$
that has the same displacement as~$\psi$. The latter coherent state is the one that maximizes the squared overlap with~$\ket{\psi}$, that is, 
\begin{align}
    \alpha(\psi)&=\arg\max\left\{|\langle\beta,\psi\rangle|^2\ |\ \beta\in\mathbb{C}\right\}\ .
    \end{align}
The triple~$\Delta(\psi):=(\Gamma(\psi),\alpha(\psi),r(\psi))$ then is what we use as a description of a pure Gaussian state~$\psi\in\mathrm{Gauss}_n$.

To represent a superposition~$\Psi=\sum_{j=1}^\chi c_j\psi_j\in\cS^{\mathrm{Gauss}_n}(\chi)$ of Gaussian states, we then simply use the collection~$\{c_j,\Delta(\psi_j)\}_{j=1}^\chi$
of corresponding  coefficients and descriptions of the individual terms in the superposition. 

To construct our classical simulation algorithms, we then argue how these descriptions can be updated under Gaussian unitaries, and how measurement outcome probabilities can be computed from these descriptions. A key building block here  (which ultimately leads to the~$O(\chi)$-runtime of the algorithm~$\Aapproximatesimulation$) is a probabilistic subroutine that estimates the norm of any vector~$\Psi\in\Sgauss_N(\chi)$ in linear time in~$\chi$.

\section{Background}

In this section, we give some background on Gaussian states and operations.
We consider the Hilbert space~$\cH_n = L^2(\bbR)^{\otimes n}$ associated with a system of~$n$ bosonic modes.

We introduce bosonic creation- and annihilation- operators, $\{a^\dagger_j\}_{j=1}^{n}$ and~$\{a_j\}_{j=1}^{n}$, which satisfy the canonical commutation relations
\begin{align}
    [a_j, a_k^\dagger] = \delta_{j,k} I 
    \quad,\quad 
    [a_j, a_k] = [a_j^\dagger, a_k^\dagger] = 0
\end{align}
for $j,k\in[n]$.
A number state is defined as 
\begin{align}
    | k \rangle = (a_1^\dagger)^{k_1} \cdots (a_n^\dagger)^{k_n} | 0 \rangle
    \quad\text{for}\quad
    k\in\bbN_0 = \bbN \cup \{0\} \ ,
\end{align}
where~$|0\rangle$ is the vacuum state defined (up to a global phase) as the state annihilated by any annihilation operator, that is
\begin{align}
    a_j | 0 \rangle = 0
    \quad\text{for all}\quad
    j \in [n] \ .
\end{align}
The set of number states~$\{|k\rangle\}_{k\in\bbN_0^n}$ is an orthonormal basis of the Hilbert space~$\cH_n$.

Instead of using creation and annihilation operators, one can work with field operators $Q_1, \ldots, Q_n, P_1, \ldots, P_n$ which are defined as
\begin{align}
    \begin{gathered}
    \begin{aligned}
    Q_j &= (a_j + a_j^\dagger) / \sqrt{2}\\
    P_j &= - i (a_j - a_j^\dagger) / \sqrt{2}
    \end{aligned}
    \end{gathered}
    \quad\text{for}\quad j\in[n]
\end{align}
and which satisfy the canonical commutation relations
\begin{align}
    \left[ R_j, R_k \right] = i \Omega_{k,l} I \ .
\end{align}
Here we defined the vector of field operators 
\begin{align}
    R = (Q_1,P_1,\ldots, Q_n, P_n) 
\end{align}
where the matrix
\begin{align}
    \Omega = 
    \bigoplus_{j=1}^n
    \begin{pmatrix}
        0 & 1 \\
        -1 & 0 
    \end{pmatrix} \ 
\end{align}
defines the symplectic form on~$\bbR^{2n}$. 

\subsection{Gaussian states}
\label{sec:gaussian-states}

The characteristic function~$\chi_\rho: \bbR^{2n} \rightarrow \bbC$ of a density operator~$\rho$ is defined as
\begin{align}
    \label{eq:chi} \chi_\rho (\xi) = \tr\left( \rho  D(\hat{d}^{-1}(\xi))  \right) 
    \quad\text{for}\quad
    \xi\in\bbR^{2n} \ ,
\end{align}
with the displacement operator
\begin{align}
    D(\alpha) = e^{i \hat{d}(\alpha)^T \Omega R} 
    \quad\text{for}\quad
    \alpha\in\bbC^n\ ,
\end{align}
 where we defined the function
\begin{align}
    \begin{gathered}
    \begin{aligned}
    \label{eq:psifunctiondescr}
    \hat{d} : \bbC^n &\rightarrow  \bbR^{2n} \\
    \alpha &\mapsto  \hat{d}(\alpha) = \sqrt{2} \left(\re(\alpha_1),\im(\alpha_1),\ldots, \re(\alpha_n),\im(\alpha_n)\right)
    \end{aligned}
    \end{gathered} \ .
\end{align}
The characteristic function~$\chi_\rho:\bbR^{2n}\rightarrow \bbC$ of a state~$\rho$ fully determines the state~$\rho$ via the Weyl transform
\begin{align}
    \label{eq:WeylTransform}
    \rho = \frac{1}{(2\pi)^{n}} \int_{\mathbb{R}^{2n}} d^{2n} \xi \, \chi_\rho(\xi) 
    D(-\hat{d}^{-1}(\xi))  \ .
\end{align}

The displacement vector~$d\in \mathbb{R}^{2n}$ of a state~$\rho \in \mathcal{B}(\cH_n)$ is the vector whose entries are the first moments of the field operators, i.e., 
\begin{align}
    \label{eq:d} d_j = \tr( \rho R_j )
    \quad\text{for}\quad
    j \in [2n] \ .
\end{align}

The covariance matrix~$ \cov \in\mathsf{Mat}_{2n\times 2n}(\mathbb{R})$ of a state~$\rho \in \cH_n$ is the matrix whose entries
\begin{align}
    \label{eq:V}   \cov_{j,k} = \tr( \rho \{ R_j - d_j I, R_k - d_k I \} )
\end{align}
for~$j,k\in[2n]$ are the second moments of the field operators. 
By definition, the covariance matrix is symmetric. It satisfies
\begin{align}
    \label{eq:Vcond}
    \Gamma + i\Omega \geq 0
\end{align}
which follows from~$\rho\geq0$.
We say a symmetric matrix~$\Gamma \in \msf{Mat}_{2n\times 2n}(\mathbb{R})$ is a valid covariance matrix if it satisfies~\cref{eq:Vcond}.

A Gaussian state~$\rho\in\cB(\cH_n)$ is a state whose characteristic function~$\chi_\rho$ has the form
\begin{align}
    \label{eq:chiGaussianState} \chi_\rho (\xi) = \exp\left(-\frac{1}{4} \xi^T  \Omega^T \cov \Omega \xi - d^T i \Omega \xi \right) 
\end{align}
with~$\xi \in \mathbb{R}^{2n}$.
The pair~$(\Gamma,d)$ are the covariance matrix and displacement of~$\rho$.
Since the characteristic function~$\chi_\rho$ of a state~$\rho$ fully determines the state~$\rho$ by~\cref{eq:WeylTransform}, a Gaussian state is fully determined by its first two moments. 
We will denote by~$\rho(\cov,d)$ a Gaussian state with covariance matrix~$\cov$ and displacement vector~$d$.

According to Williamson theorem \cite{a6cda7d0-29df-361a-a64d-8dd6c762bce6} (see also \cite[Appendix 6]{arnol2013mathematical}), any symmetric matrix~$\Gamma\in\mathsf{Mat}_{2n\times 2n}(\bbR)$ can be written as (see e.g., \cite[Theorem 8.11]{degossonSymplecticGeometryQuantum2006})
\begin{align}
    \label{eq:williamson}
    \Gamma = S \left(\bigoplus_{j=1}^n \begin{pmatrix}
        d_j & 0 \\
        0 & d_j
    \end{pmatrix}\right) S^T \ ,
\end{align}
where~$d_j\in \bbR$ for $j\in[n]$ %\in\mathsf{Mat}_{n\times n}(\bbR)$ is a diagonal matrix 
and~$S\in Sp(2n)$ where~$Sp(2n)$ denotes the group of~$2n\times 2n$ real symplectic matrices, i.e.,~$S$ satisfies~$S\Omega S^T=\Omega$. Additionally, any symplectic matrix~$S\in Sp(2n)$ can be decomposed as 
\begin{align}
    \label{eq:EulerDecomposition}
    S = K' Z  K \ ,
\end{align}
where~$K,K' \in Sp(2n) \cap O(2n)$ are symplectic orthogonal matrices
and where 
\begin{align}
    \label{eq:Zactive}
    Z = \diag(z_1, 1/z_1, \ldots, 1/z_n,  z_n)
\end{align}
with~$z_j > 1, j\in[n]$. This is called the Euler or Bloch-Messiah decomposition \cite{Arvind_1995}. Together with~\cref{eq:williamson}, the Euler decomposition allows to express a covariance matrix~$\Gamma$ of a pure state as
\begin{align}
     \label{eq:covmat-standard-form}
    \Gamma = K Z K^T
\end{align}
where~$Z$ is of the form given in~\cref{eq:Zactive} and~$K \in Sp(2n) \cap O(2n)$.

Throughout our work, we will consider the Hamiltonian
\begin{align}
    \label{eq:hamiltonian}
    H&=\sum_{j=1}^n \left( Q_j^2 + P_j^2 + I \right) 
    = 2 \sum_{j=1}^n \left( a_j^\dagger a_j + I \right) \ .
\end{align}
We note that
\begin{align}
    N=\sum_{j=1}^n a_j^\dagger a_j
\end{align}
is referred to as the number operator.
The energy of a Gaussian state~$\rho(\cov,d)$ is given by 
\begin{align}
    \label{eq:energy}
    \tr(H \rho(\cov,d)) = \frac{1}{2} \tr(\cov) + d^T d + n \ .
\end{align}
In the following, we will primarily work with pure Gaussian states. We denote by
\begin{align}
    \mathrm{Gauss}_n = \left\{ \psi\in\cH_n \ | \ |\psi\rangle \langle \psi | \text{ is a Gaussian density operator} \right\}    
\end{align}
the set of all~$n$-mode pure Gaussian states.

\begin{table}[!h]
\resizebox{\columnwidth}{!}{
\begin{tabular}{lccc}
\toprule
Gaussian unitary operation                    & unitary~$U$                                   & $S\in Sp(2n)$ & $s\in\bbR^n$ \\ \midrule 
displacement~$D(\alpha)$ &
  $e^{i\hat{d}(\alpha)^T \Omega R}$ &
  $I$ &
  $\hat{d}(\alpha)$  \\ \cmidrule(l){2-4} 
phase shift~$F_j(\phi)$           & $e^{-i\frac{\phi}{2}(Q_j^2+P_j^2-I)}$                         
&                       
$\begin{matrix}
    & \begin{matrix} Q_j \ \ \   & \ \ \  P_j \end{matrix} \\
    \begin{matrix} Q_j \\  P_j \end{matrix} 
    &
    \left(\begin{matrix}
    \cos(\phi) & \sin(\phi) \\
    -\sin(\phi) & \cos(\phi) \\
    \end{matrix}\right)
\end{matrix}$
& $0$       \\ \cmidrule(l){2-4} 
beamsplitter~$B_{j,k}(\omega)$ & $e^{ i \omega (Q_jQ_k+P_jP_k) }$ 
&                       
$\begin{matrix}
    & \begin{matrix} \ \ \ Q_j \ \ \  & \ \ \ P_j \ \ \ & \ \ \ Q_k \ \ \ & \ \ \ P_k \ \ \ \end{matrix} \\
    \begin{matrix} Q_j \\  P_j \\ Q_k \\ P_k \end{matrix} 
    &
    \left(\begin{matrix}
    \cos(\omega) & 0 & 0 & \sin(\omega) \\
    0 & \sin(\omega) & \cos(\omega)  & 0 \\
    0 &\cos(\omega) & -\sin(\omega)  & 0 \\
    -\sin(\omega) & 0 & 0 & \cos(\omega) \\
    \end{matrix}\right)
\end{matrix}$
& $0$        \\ \cmidrule(l){2-4} 
single-mode squeezing~$S_j(z)$            & $e^{i\frac{z}{2}(Q_jP_j+P_jQ_j)}$                
& 
$\begin{matrix}
    & \begin{matrix} \ Q_j   &  P_j \end{matrix} \ \\
    \begin{matrix} Q_j \\  P_j \end{matrix} 
    &
    \left(\begin{matrix}
    e^{-z} & 0 \\
    0 & e^z \\
    \end{matrix}\right)
\end{matrix}$
& $0$   \\ \bottomrule    
\end{tabular}}
\caption{\justifying Unitary operator~$U$, symplectic matrix~$S\in Sp(2n)$ and vector~$s\in\bbC^n$ associated with each Gaussian unitary operation. Apart from the case of the displacement, we only show the entries labelled by the associated field operators where~$S$ is different from the identity.}
\label{tab:gaussian-unitaries}
\end{table}

\subsection{Gaussian unitary operations}
\label{sec:gaussian-unitary-op}

An operation is Gaussian if and only if it maps Gaussian states to Gaussian states. 
Gaussian unitary operators~$U: \cB(\cH_n) \rightarrow \cB(\cH_n)$ map field operators~$R_j,j\in[2n]$ to linear combinations thereof according to
\begin{align}
    \label{eq:USprelation} U R_j U^\dagger = \sum_{j=1}^{2n} S_{j,k} R_k
    \quad\text{with}\quad
    S\in Sp(2n) \ .
\end{align}

Recall that a Gaussian state~$\rho(\cov,d)$ is determined (up to a global phase) by its covariance matrix~$\cov$ and displacement vector~$d$. This together with the fact that Gaussian operators map Gaussian states to Gaussian states means that we can describe Gaussian operations (up to a global phase) by the way they transform covariance matrices and displacement vectors. 
Concretely, a Gaussian operator~$U$ acting on~$\rho(\cov, d)$ leads to a Gaussian state~$\rho(\cov',d') = U \rho U^\dagger$ with covariance matrix~$\cov'$ and displacement vector~$d'$ given by
\begin{align}
    \label{eq:covprime} \cov' &= S \cov S^T \ , \\ 
    \label{eq:dprime} d' &= S d + s \ ,
\end{align}
where~$S\in Sp(2n)$ is related to~$U$ by~\cref{eq:USprelation} and where~$s\in \bbR^{2n}$. 
Hence, a Gaussian unitary operation~$U$ is described by a symplectic matrix~$S\in Sp(2n)$ and a vector~$s\in\bbR^{2n}$. Conversely, any pair~$(S,s)\in Sp(2n) \times \bbR^{2n}$ can be associated with a Gaussian unitary~$U$. 
The group of Gaussian unitaries is generated by
\begin{enumerate}[(i)]
\item
$n$-mode displacement operators~$D(\alpha)$, $\alpha\in\mathbb{C}^n$,
\item
single-mode phase shifters~$F_j(\phi)$, $\phi\in\mathbb{R}$
applied to a mode~$j\in[n]$,
\item
single-mode squeezing operators~$S_j(z)$, $z\geq 1$ applied to a mode~$j\in [n]$, and
\item beamsplitters~$B_{j,k}(\omega)$, $\omega\in\mathbb{R}$, $j\neq k$
applied to modes~$j,k\in[n]$.
\end{enumerate}
In~\cref{tab:gaussian-unitaries} we explicitly write these unitary operators in terms of field operators, as well as the symplectic matrix~$S\in Sp(2n)$ and vector~$s\in\bbR^{2n}$ associated with each of these operations. 

For future reference, note the displacement operator satisfies the Weyl relations
\begin{align}
    \label{eq:Wprop1} D(\alpha)^\dagger &= D(-\alpha) \ , \\ 
    \label{eq:Wprop2b} 
     D(\alpha)D(\beta) 
     &= \exp\left( i \im(\alpha^T\overline{\beta}) \right) D(\alpha+\beta)  \\
    \label{eq:Wprop2} &= \exp\left( - \hat{d}(\alpha)^T i \Omega \hat{d}(\beta)/2 \right) D(\alpha+\beta)
\end{align}
for~$\alpha,\beta\in\bbC^n$.

\subsection{Coherent  states}
\label{sec:coherent-states}

Coherent states play an important role in our algorithms and are a useful tool in many calculations. For~$\alpha\in\bbC^n$, the coherent state~$|\alpha\rangle \in \cH_n$ is the pure Gaussian state obtained by applying the displacement operation~$D(-\alpha)$ to the vacuum state, i.e., 
\begin{align}
    \label{eq:defCoherentState} 
    | \alpha \rangle = D(-\alpha) | 0 \rangle \ .
\end{align}
The coherent state~$|\alpha \rangle$ has covariance matrix and displacement vector
\begin{align}
    \cov &= I   \ ,  \\
    \label{eq:dispcoherent} \hat{d}(\alpha) &= \sqrt{2} (\re(\alpha_1), \im(\alpha_1), \dots, \re(\alpha_n), \im(\alpha_n)) \ .
\end{align}
Coherent states are eigenvalues of the annihilation operators~$a_j,j\in[n]$, i.e.,
\begin{align}
    a_j |\alpha \rangle = \alpha_j |\alpha \rangle \ ,
\end{align}
and form an overcomplete basis, satisfying the completeness relation
\begin{align}
    \label{eq:numberstaterep}
    \frac{1}{\pi^n} \int_{\bbR^{2n}} d^{2n} \alpha \, |\alpha\rangle \langle \alpha | = I_{\cH_n} \ .
\end{align}

We note that displacement operators, phase shifts and beamsplitters take coherent states to coherent states. 
More precisely, for~$\beta\in\bbC^n$, we have
\begin{align}
    \label{eq:Dcoherentstate}
    D(\beta) |\alpha\rangle = | \alpha' \rangle
\end{align}
with~$\alpha' = \exp(i \im(\alpha^T \overline{\beta}))(\alpha-\beta) \in \bbC^n$
for all~$\alpha\in\bbC^n$.
For~$\phi\in\bbR$ and~$j\in[n]$, we have
\begin{align}
    \label{eq:PScoherentstate} F_j(\phi) |\alpha \rangle = |\alpha'\rangle
\end{align}
with~$\alpha' = (\alpha_1, \ldots, \alpha_{j-1}, e^{-i\phi} \alpha_j, \alpha_{j+1}, \ldots, \alpha_n ) \in \bbC^{n}$ for all~$\alpha\in\bbC^n$ and, for~$\omega\in\bbR$ and~$j,k\in[n]$, 
\begin{align}
    \label{eq:BScoherentstate}
    B_{j,k} (\omega) | \alpha \rangle 
    = |\alpha' \rangle
\end{align}
with~$\alpha'\in\bbC^n$ the vector with entries
\begin{align}
    \label{eq:BScoherentstate-alpha}
    \alpha'_j &= \alpha_j \cos\omega - i \alpha_k \sin\omega \ , \\
    \alpha'_k &= - i \alpha_j \sin\omega + \alpha_k \cos\omega \ ,
    \quad\text{for all}\quad\alpha\in\bbC^n \ . \\
    \alpha'_\ell &= \alpha_\ell \quad\text{for}\quad \ell\notin\{j,k\} \ ,
\end{align}

\subsection{Gaussian measurements}
\label{sec:gaussian-measurements}

We consider a~$n$-mode Gaussian state~$\rho\in\cB(\cH_A \otimes \cH_B)$ on a bipartite system where~$\cH_A \simeq \cH_k$ includes~$k$ modes and~$\cH_B \simeq H_{n-k}$ includes the remaining~$n-k$ modes. 
The state~$\rho_{AB}=\rho(V,s)$ has covariance matrix
\begin{align}
    \cov = 
    \begin{pmatrix}
        \Gamma_A & \Gamma_{AB} \\ 
        \Gamma_{AB}^T & \Gamma_B
    \end{pmatrix} \ ,
\end{align}
where~$\Gamma_A\in\mathsf{Mat}_{2k\times2k}(\bbR)$ is the covariance matrix of the reduced density operator~$\rho_A$ on subsystem~$A$ and~$\Gamma_B\in\mathsf{Mat}_{2(n-k)\times 2(n-k)}(\bbR)$ is the covariance matrix of the reduced density operator~$\rho_B$ on subsystem~$B$,
and displacement vector
\begin{align}
    s=(s_A,s_B) \ ,   
\end{align}
where~$s_A\in\bbR^{2k}$ and~$s_B\in\bbR^{2(n-k)}$ are the displacements of~$\rho_A$ and~$\rho_B$, respectively. 

A heterodyne measurement of subsystem~$A$ with outcome~$\alpha\in\bbC^k$ corresponds to projecting the state~$\rho_{AB}$ onto 
\begin{align}
    \Pi_\alpha = |\alpha\rangle \langle \alpha | \otimes I^{n-k} \ .
\end{align}
The probability of obtaining outcome~$\alpha$ is given by
\begin{align}
    p(\alpha) &= \frac{1}{\pi^{k}} \tr( \Pi_\alpha \rho_{AB} )  \\
    \label{eq:heterodyne-prob} &= \frac{\exp\left( -(d(\alpha)-s_A)^T (\Gamma_A+I)^{-1} (d(\alpha)-s_A) \right) }{\pi^{k} \sqrt{\det(\Gamma_A+I)}} \ ,
\end{align}
and the post-measurement state is
\begin{align}
    \rho'_{AB} = \frac{1}{p(\alpha)} \Pi_\alpha \rho_{AB} \Pi_\alpha \ .
\end{align}
The post-measurement state~$\rho'_{AB} = \rho(\cov',d')$ is a Gaussian state with covariance matrix (see calculation e.g. in \cite[Appendix B]{garcia2007quantum})
\begin{align}
    \label{eq:heterodyne-cov}\cov' = I \oplus \left( \Gamma_B-\Gamma_{AB}^T (\Gamma_A + I)^{-1} \Gamma_{AB} \right) 
\end{align}
and displacement vector
\begin{align}
    \label{eq:heterodyne-disp} d' = (\hat{d(\alpha)}, s_B + \Gamma_{AB}^T(\Gamma_A + I)^{-1}(\hat{d}(\alpha)-s_A)) \ .
\end{align}

\subsection{Inner product formulas for Gaussian states}

The Hilbert-Schmidt inner product of density operators~$\rho,\rho'\in\cB(\cH_n)$ is connected with the~$L^2$-inner product of their characteristic functions~$\chi_\rho,\chi_{\rho'}$ by Parseval's relation
\begin{align}
    \label{eq:parseval}
    \tr\left( \rho^\dagger \rho' \right) = \frac{1}{(2\pi)^{n}}  \int_{\mathbb{R}^{2n}} d^{2n} \xi \, \overline{\chi_\rho(\xi)} \chi_{\rho'}(\xi) \ .
\end{align}
This allows us to compute the trace of two Gaussian states~$\rho(\cov_1, d_1)$ and~$\rho(\cov_2, d_2)$, with~$\cov_1,\cov_2\in\mathsf{Mat}_{2n\times2n} (\bbR)$ two covariance matrices and~$d_1, d_2\in\bbR^{2n}$. This amounts to evaluating a Gaussian integral, giving \cite{PhysRevA.61.022306,PhysRevLett.115.260501}
\begin{align}
    \label{eq:trace2gaussians}
    &\tr\left( \rho(\cov_1, d_1) \rho(\cov_2, d_2) \right) \\
    &= \frac{\exp\left( -(d_2-d_1)^T (\cov_1+\cov_2)^{-1} (d_2-d_1) \right)}{\sqrt{\det\left((\cov_1+\cov_2)/2\right)}}  \ .
\end{align}
For pure states~$\rho(\cov_1, d_1) = |\psi_1\rangle \langle\psi_1 |$ and~$\rho(\cov_2, d_2) = |\psi_2\rangle \langle\psi_2 |$ this formula gives the inner product squared~$| \langle \psi_1, \psi_2 \rangle|^2$. Notice that describing Gaussian states by their covariance matrix and displacement vector, e.g., using the characteristic function given in~\cref{eq:chiGaussianState}, does not allow computing the inner product~$ \langle \psi_1, \psi_2 \rangle$ (i.e., the inner product itself, not its absolute value). This is because the covariance matrix and displacement vector are functions of the density operator~$\rho=|\psi\rangle\langle \psi|$ (recall the definitions in~\cref{eq:d,eq:V}), not of the state vector~$|\psi\rangle$.

On the other hand, single-mode coherent states~$|\alpha\rangle,\alpha\in\bbC$ can be written explicitly in the number state basis as
\begin{align}
    \label{eq:coherent-number-basis}
    | \alpha \rangle &= \exp\left(-|\alpha|^2/2\right) \sum_{j=0}^\infty \frac{\alpha^j}{\sqrt{j!}} | j \rangle  \ .
\end{align}
It follows that the overlap between a coherent state~$|\alpha\rangle,\alpha\in\bbC^n$ and the vacuum~$|0\rangle$ is
\begin{align}
    \langle 0, \alpha\rangle = \exp\left( -\frac{1}{2} |\alpha|^2 \right) \ .
\end{align}
More generally, the overlap between two coherent states~$|\alpha_1\rangle$ and~$|\alpha_2\rangle$ for~$\alpha_1,\alpha_2\in\bbC^n$ is
\begin{align}
\label{eq:overlapCoherentSt}
    \langle \alpha_1, \alpha_2 \rangle &= \exp\left( -\frac{1}{2}|\alpha_1|^2 - \frac{1}{2}|\alpha_2|^2 + \overline{\alpha_1}^T \alpha_2  \right) \ .
\end{align}

\section{Tracking relative phases in linear optics}
\label{sec:keeptrackphases}

In this section, we establish basic building blocks for our algorithms. In particular, we provide the following:
\begin{enumerate}[(i)]
    \item\label{it:desc} A classical description of a Gaussian state~$|\psi\rangle\in\mathrm{Gauss}_n$ which includes information about its phase. 
    \item\label{it:alg-overlap} A polynomial-time classical algorithm~$\Aoverlap$ which computes the overlap (including the phase, not just its absolute value) between any two Gaussian states from their classical descriptions.
    \item\label{it:alg-evolution} A polynomial-time classical algorithm~$\Aunitary$ which updates the classical description of a Gaussian state when applying a Gaussian unitary operation (a displacement, a phase shit, a beamsplitter or a single-mode squeezing operation).
    \item\label{it:alg-evolution} A polynomial-time classical algorithm~$\Ameasure$ which updates the classical description of a Gaussian state when applying a heterodyne measurement.
\end{enumerate}

In~\cref{sec:extended-description-gaussian} we address Item~\eqref{it:desc}, in~\cref{sec:overlap} we address Item~\eqref{it:alg-overlap} and in~\cref{sec:algorithms-gaussian} we address Item~\eqref{it:alg-evolution}.

\subsection{Extended description of a Gaussian state}
\label{sec:extended-description-gaussian}

Here we give a description of a Gaussian state~$|\psi\rangle \in \mathrm{Gauss}_n$ which includes information about its phase. 
While this phase is irrelevant when considering a single state, it will be important when considering superpositions, as relative phases are physically observable.
Our proposal is to fix the phase of the state~$|\psi\rangle$ relative to a reference state with which~$|\psi\rangle$ has non-zero overlap. We use coherent states~$|\alpha\rangle,\alpha\in\bbC^n$ as reference states. The overlap~$|\langle \alpha, \psi \rangle|^2$ is maximal when choosing as the reference state the coherent state~$|\alpha\rangle$ with the same displacement vector as the state~$|\psi\rangle$.

\begin{definition}[Extended description of a Gaussian state]
    \label{def:description}
    Let~$|\psi\rangle \in \mathrm{Gauss}_n$ be a Gaussian state. We call a tuple
    \begin{align}
        \desc = \left( \cov, \alpha,  r  \right) \in \mathsf{Mat}_{2n\times 2n}(\bbR) \times \bbC^{n} \times \bbC 
    \end{align}
    the description of the state~$|\psi\rangle$ if the following two conditions hold:
    \begin{enumerate}[(i)]
        \item \label{it:DefDescCond1} $\cov = \cov( \psi )$ and~$\hat{d}(\alpha)=d(\psi)$ are respectively the covariance matrix and the displacement vector of~$|\psi\rangle$.
        \item \label{it:DefDescCond3}
        $r = \langle \alpha, \psi \rangle$, i.e., $r$ is the overlap of the state~$|\psi\rangle$ with the coherent state~$|\alpha\rangle$ that has the same displacement.
    \end{enumerate}
\end{definition}

We denote by~$\mathsf{Desc}_n$ the set of all descriptions of~$n$-mode pure Gaussian states. A description in~$\mathsf{Desc}_n$ uniquely fixes a Gaussian state and each Gaussian state~$|\psi\rangle \in\mathrm{Gauss}_n$ admits a unique description. Hence, we have a bijective function 
\begin{align}
    \begin{gathered}
    \begin{aligned}
    \label{eq:psifunctiondescr}
    \psi  :\quad \mathsf{Desc}_n \quad&\rightarrow\quad  \mathrm{Gauss}_n\\
    \Delta \quad&\mapsto\quad  \psi(\Delta)
    \end{aligned}
    \end{gathered} \ .
\end{align}

We establish a lower bound on the absolute value of the squared overlap~$|r|^2$ in~\cref{def:description} as a corollary of the following lemma.
We define
\begin{align}
    \var_\psi(R_j) = \langle \psi, R_j^2 \psi\rangle - \langle \psi, R_j \psi\rangle^2 
    \quad\text{for}\quad j\in[2n] \ .
\end{align}

\begin{lemma}
\label{lem:fidelitybound}
Let~$|\psi\rangle \in \mathrm{Gauss}_n$ be a pure~$n$-mode Gaussian state. Then
\begin{align}
\label{eq:fidelitybound}
&\sup_{\alpha\in\mathbb{C}^n}  |\langle\alpha,\psi\rangle|^2 = \max_{\alpha\in\mathbb{C}^n} 
|\langle\alpha,\psi\rangle|^2 \\
&\qquad\geq \frac{1}{\sqrt{ \frac{1}{2} \sum_{j=1}^{n} \left( \var_{\psi}(Q_j) + \var_{\psi}(P_j) - 1 \right) + 1}} \ .
\end{align}
\end{lemma}
\begin{proof}
See~\cref{app:proof-r-lowerbound}.
\end{proof}

The field operators satisfy Heisenberg's uncertainty relation
\begin{align}
    \label{eq:uncertaintyQP} 
    \var_\psi(Q_j) \cdot \var_\psi(P_j) &\geq  \frac{1}{4} 
    \quad\text{for}\quad
    j\in[n] \ .
\end{align}
This together with the arithmetic and geometric means inequality
\begin{align}
   \frac{1}{2}(a+b) \geq \sqrt{ab}
   \quad\text{for}\quad
   a,b\in\bbR
\end{align}
gives
\begin{align}
    \var_\psi(Q_j) + \var_\psi(P_j) \geq 1 
    \quad\text{for}\quad
    j\in[n] \ .
\end{align}
It follows that the rhs of~\cref{eq:fidelitybound} in~\cref{lem:fidelitybound} is at most 1, with equality when~$|\psi\rangle$ is a coherent state, a state of minimal uncertainty.

\begin{corollary}
    \label{cor:r-lowerbound}
    Consider the description~$\Delta = (\Gamma,\alpha,r) \in  \mathsf{Mat}_{2n\times 2n}(\bbR) \times \bbC^{n} \times \bbC$ of a Gaussian state~$|\psi\rangle\in \mathrm{Gauss}_n$. Then,
    \begin{align}
        \label{eq:r-lowerbound}
        |r|^2 \geq \frac{1}{\sqrt{ \frac{1}{2} \sum_{j=1}^{n} \left( \var_{\psi}(Q_j) + \var_{\psi}(P_j) - 1 \right) + 1}} \ .
    \end{align}
    In particular, for~$|\psi\rangle$ a centered state we have
    \begin{align}
        \label{eq:r-lowerbound-centered}
        |r|^2 \geq \frac{1}{\sqrt{N+1}} \ .
    \end{align}
\end{corollary}

\begin{proof}
    We have
    \begin{align}
        r = \langle \alpha, \psi \rangle = 
        \max_{\beta\in\mathbb{C}^n} |\langle\beta,\psi\rangle|^2 \ ,
    \end{align}
    see the proof of~\cref{cor:r-lowerbound} in~\cref{app:proof-r-lowerbound}. The claim in~\cref{eq:r-lowerbound} then follows from~\cref{lem:fidelitybound}.

    Consider that~$|\psi\rangle$ is centered, i.e.,~$d(\psi)=0$. Then, we have
    \begin{align}
        \sum_{j=1}^{n} \left( \var_{\psi}(Q_j) + \var_{\psi}(P_j) - 1 \right) 
        &= \langle \psi, H\psi \rangle - 2n = 2N
    \end{align}
    and the claim in~\cref{eq:r-lowerbound-centered} follows.
\end{proof}

The description of a pure Gaussian state~$|\psi\rangle \in \mathrm{Gauss}_n$ given in~\cref{def:description} is a minimal extension of the standard description~$(\Gamma(\psi), d(\psi))$ in terms of the covariance matrix and displacement vector. Additionally, we simply give the overlap~$r$ with the chosen reference coherent state which suffices to fix the phase of~$|\psi\rangle$. 
The overlap~$r$ is non-zero for states with bounded variance / occupation number (see~\cref{cor:r-lowerbound}), in which case our algorithms work with finite-precision arithmetic.\newline

\subsection{Computing the overlap between Gaussian states}
\label{sec:overlap}

In this section we give an algorithm~$\Aoverlap$ which computes overlaps between Gaussian states. 
This algorithm uses an expression given in~\cref{prop:trace3statesWWeylPsi} which relates inner products between three Gaussian states and a displacement operator~$D(\alpha),\alpha\in\bbC^n$. This is a generalization of a formula derived in Ref.~\cite{PhysRevA.61.022306} for the case of centered states and~$\alpha$ set to zero.

\begin{lemma}
    \label{prop:trace3statesWWeylPsi}
    Consider three pure Gaussian states~$\psi_j \in\mathrm{Gauss}_n,j\in[3]$ with respective covariance matrices~$\cov_j\in\mathsf{Mat}_{2n\times 2n}(\mathbb{R}),j\in[3]$ and displacement vectors~$d_j\in\bbR^{2n},j\in[3]$. Let~$\alpha \in \bbC^{n}$. Then
    \begin{align}
        &\langle \psi_3, D(\alpha) \psi_1 \rangle \langle \psi_1, \psi_2 \rangle \langle \psi_2, \psi_3 \rangle\\
        \label{eq:trace3statesWWeylPsi} &= \frac{\exp\left( - {d_1'}^T W_1 {d_1'} - {d_2'}^T W_2 {d_2'} + {d_1'}^T W_3 d_2' - \hat{d}(\alpha)^T \Omega^T \cov_5 \Omega \hat{d}(\alpha) - i\hat{d}(\alpha)^T \Omega^T ( W_4 d_1'  + W_5 d_2' + d_3 ) \right)}{\sqrt{\det(\Omega^T(\cov_2 +\cov_3 )\Omega/2) \det(\Omega^T(\cov_1 +\cov_4 )\Omega/2)}}
    \end{align}
    with
    \begin{align}
        d_1' &= d_1 - d_3 \ , \\ 
        d_2' &= d_2 - d_3 \ , \\ 
        \cov_4 &= \cov_3 - (\cov_3+i\Omega)(\cov_2+\cov_3)^{-1} (\cov_3-i\Omega) \ , \\
        \label{eq:Gamma5}\cov_5 &= \frac{1}{4} \cov_1 - \frac{1}{4} (\cov_1 -i\Omega) (\cov_1  + \cov_4 )^{-1} (\cov_1  + i\Omega) \ , \\
        \label{eq:W1}W_1 &= (\cov_1  + \cov_4 )^{-1} \ , \\
        \label{eq:W2}W_2 &= (\cov_2 +\cov_3 )^{-1} + (\cov_2 +\cov_3 )^{-1} (\cov_3 -i\Omega) (\cov_1  + \cov_4 )^{-1} (\cov_3 +i\Omega) (\cov_2 +\cov_3 )^{-1} \ , \\
        \label{eq:W3}W_3 &= 2 (\cov_1  + \cov_4 )^{-1} (\cov_3 +i\Omega)  (\cov_2 +\cov_3 )^{-1} \ , \\
        \label{eq:W4}W_4 &= I - (\cov_1  - i\Omega) (\cov_1  + \cov_4 )^{-1} \ , \\
        \label{eq:W5} W_5 &= (\cov_1  - i\Omega) (\cov_1  + \cov_4 )^{-1} (\cov_3 +i\Omega) (\cov_2 +\cov_3 )^{-1} \ .
\end{align}
\end{lemma}

\begin{proof}
    See~\cref{app:proof-trace3statesWWeylPsi}.
\end{proof}

For convenience, we introduce an algorithmic routine~$\Aoverlaptriple$ which evaluates~\cref{eq:trace3statesWWeylPsi} in~\cref{prop:trace3statesWWeylPsi}. This routine is one of the main ingredients in the algorithms~$\Aoverlap$ and $\Ameasure$.
The algorithm~$\Aoverlaptriple$ takes covariance matrices~$\cov_1, \cov_1, \cov_2$ and displacement vector~$d_1,d_2,d_3$ of three Gaussian states~$\psi_1, \psi_2, \psi_3$ as well as overlaps~$u = \langle \psi_3, D(\lambda) \psi_1 \rangle, v = \langle \psi_1, \psi_2 \rangle$ and computes the overlap~$\langle \psi_2, \psi_3\rangle$. It has runtime~$O(n^3)$ as it consists of basic linear algebra operations. We give pseudocode for the algorithm. 

A diagrammatic illustration of what the subroutine~$\Aoverlaptriple$ achieves is given in~\cref{fig:overlaptriple}. These diagrams will be useful to construct/analyze our main algorithms.

\begin{figure}[!h]
\raggedright
\par\rule{\columnwidth}{1pt}
\textbf{Algorithm~$\Aoverlaptriple$}
\vspace{-2mm}
\par\rule{\columnwidth}{1pt}
\textbf{Input: }{$\cov_1,\cov_2,\cov_3,d_1, d_2, d_3,u,v,\lambda$} \\
where $
\begin{cases}
\text{$\cov_j \in \msf{Mat}_{2n\times2n}(\bbR)$ covariance matrix of} \\
\qquad\text{a pure Gaussian state~$|\psi_j\rangle$ for~$j\in[3]$} \\
\text{$d_j \in \bbR^{2n}$ for~$j\in[3]$} \\
\text{$\lambda\in \bbC^{n}$ such that~$\langle \psi_3, D(\lambda) \psi_1\rangle \neq 0$} \\
\text{$u=\langle \psi_3, D(\lambda) \psi_1\rangle $ and~$v=\langle \psi_1, \psi_2\rangle\neq 0$} 
\end{cases}$ \\
\textbf{Output: }{$o\in\bbC$} 
\begin{algorithmic}[1]
         \State{$T \leftarrow \textrm{ evaluate the rhs of~\cref{eq:trace3statesWWeylPsi} in~\cref{prop:trace3statesWWeylPsi}}$}
         \State{$o\leftarrow u^{-1} v^{-1} T$}
         \Comment{compute the overlap~$\langle \psi_2, \psi_3 \rangle$}
         \State{\textbf{return}~$o$} 
\end{algorithmic}
\vspace{-3mm}
\par\rule{\columnwidth}{1pt}
%\label{fig:Aoverlaptriple}
\end{figure}

\begin{figure}[!b]
\begin{subfigure}[t]{0.485\textwidth}
    \centering
    \includegraphics[height=5cm]{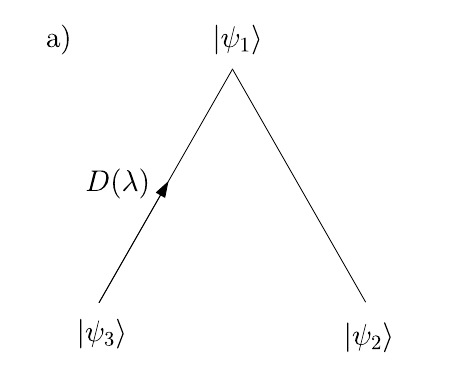}
    %\label{fig:overlaptriple-a}
\end{subfigure}
\hfill
\begin{subfigure}[t]{0.485\textwidth}
\centering
    \includegraphics[height=5cm]{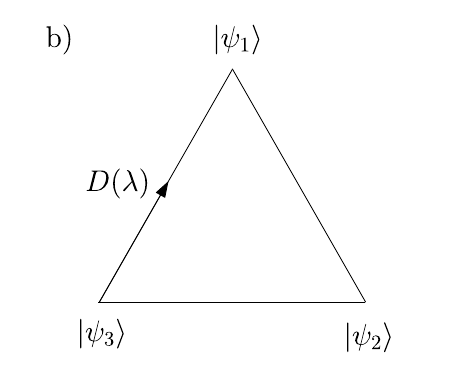}
    %\label{fig:overlaptriple-b}
\end{subfigure}
\caption{\justifying A graphical representation of the algorithm~$\Aoverlaptriple$. Solid lines represent inner products between the states at the vertices that are given / have been computed and are non-zero. Inner products of the form~$\langle \psi_3, D(\lambda) \psi_1 \rangle$ are represented by arrows. a) The input to the algorithm~$\Aoverlaptriple$ consists of covariance matrices and displacements of three Gaussian states~$\psi_1,\psi_2,\psi_3$ together with non-zero overlaps~$u = \langle \psi_3, D(\lambda) \psi_1 \rangle, v = \langle \psi_1, \psi_2 \rangle$. b) Applying~$\Aoverlaptriple$ provides the inner product~$\langle\psi_2,\psi_3\rangle$. In this diagrammatic representation, this completes the triangle with vertices~$\psi_1,\psi_2,\psi_3$.}
\label{fig:overlaptriple}
\end{figure}

We give an efficient algorithm~$\Aoverlap$ which computes inner products~$\langle \psi(\Delta_1),\psi(\Delta_2)\rangle$ from descriptions~$(\Delta_1,\Delta_2)\in\Desc_n \times \Desc_n$.
We give pseudocode for the algorithm, illustrate it in~\cref{fig:Aoverlap} and show the associated claims in~\cref{lem:overlap}.

\begin{figure}[H]
\raggedright
\par\rule{\columnwidth}{1pt}
\textbf{ Algorithm $\Aoverlap$}
\vspace{-2mm}
\par\rule{\columnwidth}{1pt}
\textbf{Input: }{$\desc_1, \desc_2$} \\
where $
    \desc_1=(\cov_1,\alpha_1,r_1), \desc_2=(\cov_2,\alpha_2,r_2)\in\Desc_n
$ \\
\textbf{Output: }{$w \in \bbC$} 
\begin{algorithmic}[1] 
    \State{\label{alg:triplematricesAconvertnew}$(\cov'_1,\cov'_2,\cov'_3)\leftarrow (I, \cov_1, \cov_2)$}
    
    \Comment{covariance matrices of~$(|\alpha_1\rangle, \psi(\desc_1), \psi(\desc_2))$}
    \State{\label{alg:tripledispvec}$(d'_1, d'_2, d'_3) \leftarrow (\hat{d}(\alpha_1), \hat{d}(\alpha_1), \hat{d}(\alpha_2))$}
    
    \Comment{displacement vectors of~$(|\alpha_1\rangle, \psi(\desc_1), \psi(\desc_2))$}
    \State{\label{it:OalgGamma}$\lambda\leftarrow \alpha_1-\alpha_2$}
    \State{\label{alg:OVu}$u\leftarrow e^{-i \im(\alpha_1^T \overline{\alpha_2})} \overline{r_2} $}
    \Comment{$u=\langle  \psi(\Delta_2), D(\lambda) \alpha_1 \rangle$}
    \State{\label{alg:OVv}$v\leftarrow r_1$}
    \Comment{$v=\langle \alpha_1, \psi(\desc_1) \rangle$}
    \State{\label{alg:OVw}$w\leftarrow \Aoverlaptriple(\cov'_1,\cov'_2,\cov'_3, d'_1, d'_2, d'_3, u, v, \lambda)$} 
    
    \Comment{compute ~$\langle \psi(\desc_1), \psi(\desc_2) \rangle$}
    \State \textbf{return}~$w$ 
\end{algorithmic}
\vspace{-3mm}
\par\rule{\columnwidth}{1pt}
%\label{alg:Aoverlap}}
\end{figure}

\begin{theorem}
\label{lem:overlap}
The algorithm~$\Aoverlap:\Desc_n\times \Desc_n\rightarrow\mathbb{C}$ runs in time~$O(n^3)$. It satisfies
\begin{align}
\Aoverlap(\desc_1,\desc_2)&=\langle \psi(\desc_1),\psi(\desc_2)\rangle \label{eq:overlapidentitydonedtwo}
\end{align}
for all~$ \desc_1,\desc_2\in\Desc_n\ $.
\end{theorem}

\begin{proof}
Let~$\desc_j=(\cov_j,\alpha_j,r_j)\in\Desc_n$ for~$j\in[2]$. 
Consider the triple of pure states
\begin{align}
(\psi_1,\psi_2,\psi_3)&=(|\alpha_1\rangle,\psi(\Delta_1),\psi(\Delta_2))\ .
\end{align}
The matrices~$\cov_j', j\in[3]$ in Line~\ref{alg:triplematricesAconvertnew} are the respective covariance matrices of~$\psi_j, j\in [3]$. Recall that~$\hat{d}(\alpha)$ defined in~\cref{eq:dispcoherent} is the displacement vector of the coherent state~$|\alpha\rangle$ and of a Gaussian state with description~$(\Gamma,\alpha,r)$.
Then, the vectors~$d_j', j\in[3]$ in Line~\ref{alg:tripledispvec} are the respective displacement vectors of~$\psi_j, j\in [3]$. 

Using the Weyl relations given in~\cref{eq:Wprop2b} we observe that the vector~$\lambda=\alpha_1-\alpha_2$ in Line~\ref{it:OalgGamma} is such that
\begin{align}
    D(\lambda) | \alpha_1 \rangle &=  D(\alpha_1-\alpha_2) | \alpha_1 \rangle \\ 
    &= e^{-i \im(\alpha_1^T \overline{\alpha_2})}
    |\alpha_2 \rangle \ .
\end{align}
Then, in Line~\ref{alg:OVu} we have 
\begin{align}
    u &= e^{-i\im(\alpha_1^T \overline{\alpha_2})} \overline{r_2} \\
    &= e^{-i\im(\alpha_1^T \overline{\alpha_2})} 
    \langle \psi(\Delta_2), \alpha_2  \rangle \\
    &= \langle \psi(\Delta_2) , D(\lambda) \alpha_1 \rangle \\
    &= \langle \psi_3, D(\lambda) \psi_1 \rangle 
\end{align}
while in Line~\ref{alg:OVv}  we have  
\begin{align}
    v &= r_1 \\
    &= \langle \alpha_1, \psi(\Delta_1) \rangle\\
    &= \langle \psi_1, \psi_2 \rangle \ .
\end{align}
It follows from the properties of~$\Aoverlaptriple$ that 
the output~$w$ computed in Line~\ref{alg:OVw} is 
\begin{align}
w&=\langle \psi(\desc_1),\psi(\desc_2)\rangle\ .
\end{align}
This implies the claim.

The runtime of~$\Aoverlap$ is~$O(n^3)$ due to the use of the subroutine~$\Aoverlaptriple$.
\end{proof}

\begin{figure}[H]
\centering
\begin{subfigure}[t]{0.315\textwidth}
\centering
    \includegraphics[height=4.2cm]{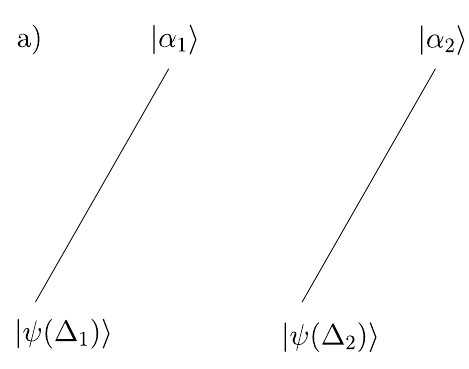}
    % \label{fig:overlap-a}
\end{subfigure}
\hfill
\begin{subfigure}[t]{0.315\textwidth}
\centering
    \includegraphics[height=4.2cm]{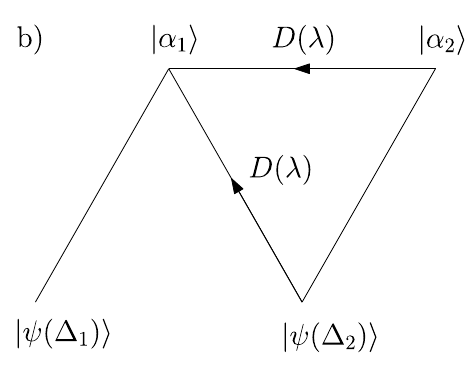}
    %\label{fig:overlap-b}
\end{subfigure}
\hfill
\begin{subfigure}[t]{0.315\textwidth}
\centering
    \includegraphics[height=4.2cm]{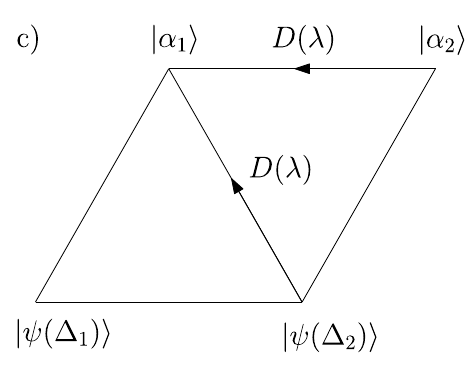}
    % \label{fig:overlap-d}
\end{subfigure}
\caption{An illustration of the algorithm~$\Aoverlap$. a) The input to~$\Aoverlap$ consists of descriptions of two Gaussian states $\psi(\Delta_1),\psi(\Delta_2)$ which includes $\alpha_1,\alpha_2\in\bbC^n$ and the overlaps~$r_j=\langle \alpha_j,\psi(\Delta_j)\rangle, j\in [2]$. b) The algorithm computes~$\lambda\in\bbC^n$ and~$\vartheta\in\bbR$ such that~$D(\lambda) |\alpha_1 \rangle = e^{i\vartheta} |\alpha_2\rangle$. This implies that~$\langle \alpha_2, D(\lambda) \alpha_1  \rangle = e^{i\vartheta}$ and~$\langle \psi(\Delta_2), D(\lambda)\alpha_1 \rangle = e^{i\vartheta} \overline{r_2}$ are known. c) Lastly, the subroutine~$\Aoverlaptriple$ is applied to complete a triangle. This amounts to computing the inner product~$w=\langle \psi(\Delta_1),\psi(\Delta_2) \rangle$ which the algorithm returns.}
\label{fig:Aoverlap}
\end{figure}

\subsection{Computing descriptions of evolved states}
\label{sec:algorithms-gaussian}

We give an algorithm called~$\Aunitary$ which, given the description~$\Delta$ of a Gaussian state~$\psi$ and a description~$\Delta_U$ of a Gaussian unitary~$U$ identified by a label~$\ell\in\{\mathsf{displacement}, \mathsf{phaseshift},\allowbreak \mathsf{beamsplitter}, \mathsf{squeezing}\}$, computes the description of the state~$U|\psi\rangle$.
We define a description for each Gaussian unitary operation considered:
\begin{enumerate}
\item $\alpha \in \bbC^n$ is a description of the displacement~$D(\alpha)$,
\item $(\phi, j) \in \bbR \times [n]$ is a description of the phase shifter~$F_j(\phi)$,
\item $(\omega,j,k) \in \bbR \times [n] \times [n]$
is a description of the beamsplitter~$B_{j,k}(\omega)$, and
\item $(z,j) \in (0,+\infty) \times [n]$
is a description of the single-mode squeezing operator~$S_{j}(z)$.
\end{enumerate}

The algorithm~$\Aunitary$ updates the covariance matrix and displacement vector of the state as prescribed by standard methods (see~\cref{sec:gaussian-unitary-op}). The novelty is in updating the overlap between the state and its reference state (recall this fixes the phase of the state).

Displacements, phase shifters and beamsplitters map coherent states to coherent states (see~\cref{sec:coherent-states}). Because of this, updating the description of a Gaussian state evolved with one of these operations is done similarly and it relies on the fact that the new reference state is the evolved reference state. 

Updating a description of a Gaussian state after a single-mode squeezing operation is more involved. Unlike the previous operations, single-mode squeezing operations do not map coherent states to coherent states. 
Instead, we have
\begin{align}
    \label{eq:squeezedcoherentstate} S_j(z) |\alpha\rangle = D(-\alpha') S_j(z) |0\rangle \ ,
\end{align}
for~$j\in[n], z\in\bbR, \alpha \in \bbC^{n}$, where~$\alpha'\in\bbC^n$ has entries
\begin{align}
    \label{eq:alphap1} \alpha'_j &= \alpha_j \cosh z - \overline{\alpha_j} \sinh z \ , \\
    \label{eq:alphap2} \alpha'_k &= \alpha_k 
    \quad\text{for}\quad k \neq j \ .
\end{align}
Because of this, we build an algorithmic subroutine called~$\Asqueezing$ which, given the description~$\Delta\in\Desc_n$ of a Gaussian state~$|\psi\rangle$, a parameter~$z\in(0,+\infty]$ and~$j\in[n]$, returns the description of the squeezed state~$S_j(z) | \psi(\Delta) \rangle$. We call this subroutine in the algorithm~$\Aunitary$ when~$U$ is a single-mode squeezing operator.

First, we give the subroutine~$\Asqueezing$, including pseudocode and associated claims in~\cref{lem:Asqueezing}, as well as a pictorial description of the algorithm in~\cref{fig:squeezing}.

\begin{figure}[H]
\raggedright
\par\rule{\columnwidth}{1pt}
\textbf{Algorithm~$\Asqueezing$}
\vspace{-2mm}
\par\rule{\columnwidth}{1pt}
\textbf{Input: }{$\desc, z, j$} \\
where
$ 
\begin{cases}
    \desc=(\cov,\alpha,r) \in \Desc_n \\
    z\in(0,+\infty) \\
    j\in[n]
\end{cases} 
$ \\
\textbf{Output: }{$(\cov', \alpha', r') \in \Desc_n$} 
\begin{algorithmic}[1]
    \State{\label{it:AlgSQmatS}$S\leftarrow S$ given in the `single-mode squeezing' row of~\cref{tab:gaussian-unitaries}}
    \State{\label{it:algSQcovP}$\cov' \leftarrow S \Gamma S^T$}
    \Comment{covariance matrix of~$S_j(z)|\psi\rangle$}
    \State{\label{it:algSQcovP2}$\cov'' \leftarrow S S^T$}
    \Comment{covariance matrix of~$S_j(z)|\alpha\rangle$}
    
    \State{\label{eq:SQalgAlphaP1}$\alpha'\leftarrow \alpha$}
    \State{\label{eq:SQalgAlphaP2}$\alpha'_j \leftarrow \alpha_j \cosh(z) - \overline{\alpha_j} \sinh(z)  $}
    
    \State{\label{it:algSQddd}$d_1,d_2,d_3 \leftarrow \hat{d}(\alpha')$}
    \Comment{displacement vectors of~$S_j(z)|\alpha\rangle, |\alpha'\rangle, S_j(z)|\psi\rangle$}

    \State{\label{it:algSQGGG}$(\Gamma_1,\Gamma_2,\Gamma_3) \leftarrow (\Gamma'',I,\Gamma')$}
    \Comment{covariance matrices of~$(S_j(z)|\alpha\rangle, |\alpha\rangle, S_j(z)|\psi\rangle)$}

    \State{\label{it:algSQu}$u \leftarrow \overline{r}$}
    \Comment{compute~$\langle S_j(z) \psi, S_j(z) \alpha \rangle$}
    \State{\label{it:algSQv}$v \leftarrow  1/\sqrt{ \cosh(z) }$}
    \Comment{compute~$\langle S_j(z) \alpha, \alpha' \rangle$}
    \State{\label{it:algSQrp}$r'\leftarrow \Aoverlaptriple(\Gamma_1, \Gamma_2, \Gamma_3, d_1,d_2,d_3, u, v, 0)$}
    
    \Comment{compute~$\langle \alpha', S_j(z) \psi \rangle$}
    \State{\textbf{return}~$(\cov', \alpha', r')$ }
\end{algorithmic}
\vspace{-3mm}
\par\rule{\columnwidth}{1pt}
%\label{alg:Asqueezing}
\end{figure}

\begin{lemma}
    \label{lem:Asqueezing}
    The algorithm~$\Asqueezing: \Desc_n \times (0,+\infty) \times [n] \rightarrow \Desc_n $ runs in time~$O(n^3)$. The algorithm is such that
    \begin{align}
        \label{eq:propsqueezing1}
        |\psi( \Asqueezing(\desc, z, j) )\rangle = S_j(z) | \psi(\desc) \rangle
    \end{align}
    for all~$\desc \in \Desc_n, z\in\bbR$ and~$j \in [n]$.
    That is, the output of the algorithm is a description of the state~$S_j(z) |\psi(\Delta) \rangle$.
\end{lemma}

\begin{proof}
    The algorithm~$\Asqueezing$ returns a tuple \begin{align}
        \Delta' = (\Gamma', \alpha', r') \in \msf{Mat}_{2n\times 2n}(\bbR) \times \bbC^n \times \bbC \ .    
    \end{align}
    Let us show that the tuple~$\Delta'$ is the description of the evolved state~$S_j(z) | \psi(\Delta) \rangle$ for any~$\Delta \in \Desc_n$, $z \in (0,+\infty)$ and $j\in[n]$. This corresponds to showing that~$\Delta'$ satisfies Items~\eqref{it:DefDescCond1} and~\eqref{it:DefDescCond3} in~\cref{def:description}.
    Consider that
    \begin{align}
        \Delta = (\Gamma,\alpha,r) \in \msf{Mat}_{2n\times 2n}(\bbR) \times \bbC^n \times \bbC \ ,
    \end{align}
    i.e., $\psi(\Delta)$ has covariance matrix~$\Gamma$, displacement~$\hat{d}(\alpha)$ and~$r$ is such that~$r=\langle \alpha, \psi(\Delta)\rangle$.

\begin{figure}[H]
%\centering
\begin{subfigure}[t]{0.315\textwidth}
%\centering
    \includegraphics[height=5cm]{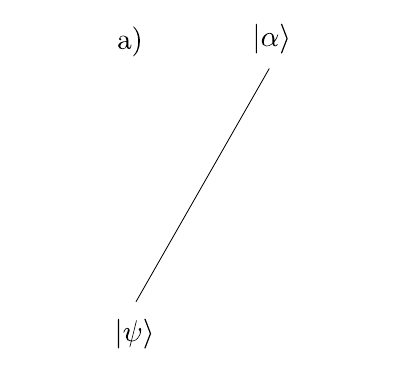}
    % \label{fig:squeezing-a}
\end{subfigure}
\hfill
\begin{subfigure}[t]{0.315\textwidth}
\centering
    \includegraphics[height=5cm]{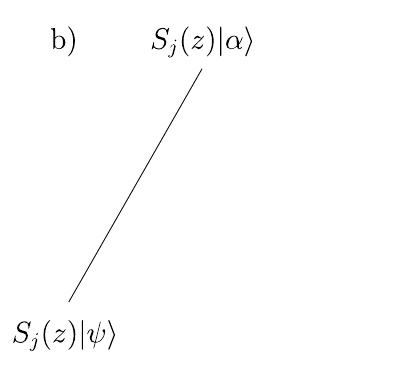}
    %\label{fig:squeezing-b}
\end{subfigure}
\hfill
\begin{subfigure}[t]{0.315\textwidth}
\centering
    \includegraphics[height=5cm]{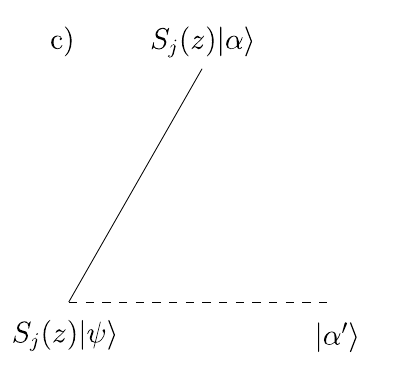}
    %\label{fig:squeezing-c}
\end{subfigure}
\hfill \\
\begin{subfigure}[t]{0.315\textwidth}
\centering
    \includegraphics[height=5cm]{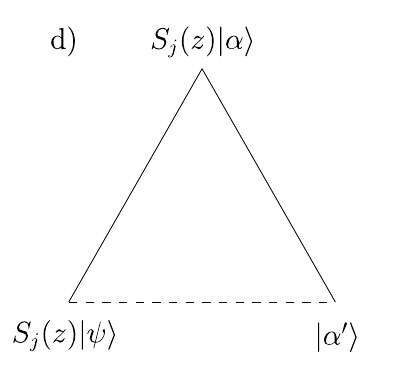}
    % \label{fig:squeezing-d}
\end{subfigure}
\hfill
\begin{subfigure}[t]{0.315\textwidth}
\centering
    \includegraphics[height=5cm]{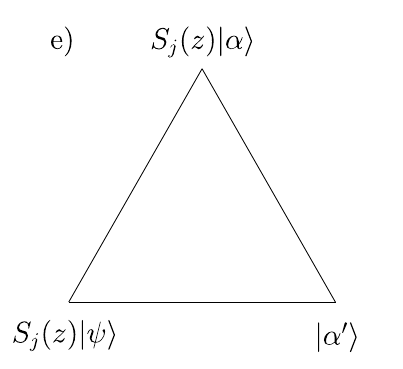}
    %\label{fig:squeezing-e}
\end{subfigure}
\begin{subfigure}[t]{0.315\textwidth}
\hfill
\end{subfigure}
\hfill
\caption{An illustration of the algorithm~$\Asqueezing$. Dashed lines correspond to inner products between the states at the vertices which are known to be non-zero but have not been computed. a) The input to the algorithm~$\Asqueezing$ is the description of a Gaussian state~$\psi$ which includes the overlap~$r = \langle \alpha, \psi\rangle$, as well as~$z\in\bbR$ and an index $j\in[n]$ associated to a single-mode squeezing operation~$S_j(z)$. b) By unitarity of~$S_j(z)$, the input data provides the inner product~$\langle S_j(z) \alpha, S_j(z) \psi \rangle = \langle \alpha, \psi \rangle = r$. c) The reference state of the evolved state~$S_j(z)|\psi\rangle$ is the coherent state~$|\alpha'\rangle$ with $\alpha'\in\bbC^n$ such that~$S(z) |\alpha\rangle = D(\alpha') S(z) | 0 \rangle$. This is due to~$S_j(z)|\psi\rangle$ having the same displacement vector~$d(\alpha')$ as~$S_j(z)|\alpha\rangle$. d) The inner product~$\langle \alpha', S_j(z) \alpha \rangle$ is computed using~\cref{eq:squeezedcoherentstateCoherentStOverlap}. e) In the last step, the subroutine~$\Aoverlaptriple$ is used to compute~$r'=\langle\alpha', S_j(z) \psi\rangle$. Thus,~$(\cov', \alpha', r')$ is the description of~$S_j(z) |\psi\rangle$.}
\label{fig:squeezing}
\end{figure}

Recall from~\cref{sec:gaussian-unitary-op} that
\begin{align}
    \Gamma' = S \Gamma S^T    
\end{align}
computed in Line~\ref{it:algSQcovP} is the covariance matrix of~$S_j(z) | \psi(\Delta) \rangle$, where~$S$ in Line~\ref{it:AlgSQmatS} is the symplectic matrix associated to~$S_j(z)$. 
We have
\begin{align}
    S_j(z) |\alpha\rangle = D(-\alpha') S_j(z) |0\rangle \ ,
\end{align}
where~$\alpha'$ is the vector given in~\cref{eq:alphap1,eq:alphap2} and initialized in Lines~\ref{eq:SQalgAlphaP1} and~\ref{eq:SQalgAlphaP2}. 
The state~$ S_j(z) |\alpha\rangle$ has displacement vector $\hat{d}(\alpha')$ and covariance matrix $\Gamma'' = SS^T$.
Since~$|\alpha\rangle$ and~$|\psi(\Delta)\rangle$ have the same displacement (by the definition of~$\Delta$) and~$S_j(z)$ maps states with the same displacement to states with the same displacement, then the state~$S_j(z)|\psi(\Delta) \rangle$ has the same displacement~$\hat{d}(\alpha')$ as~$| \alpha' \rangle$. Then, $\Delta'$ satisfies Item~\eqref{it:DefDescCond1} in~\cref{def:description}.

It remains to show that the description~$\Delta'$ satisfies Item~\eqref{it:DefDescCond3} in~\cref{def:description}. This corresponds to showing that~$r'=\langle \alpha, S_j(z) \psi \rangle$.
In Lines~\ref{it:algSQddd} and~\ref{it:algSQGGG} we set 
\begin{align}
    d_j &= \hat{d}(\alpha') 
    \quad\text{for}\quad
    j\in[3]\ , \\
    (\Gamma_1, \Gamma_2, \Gamma_3) &= (\Gamma'',I,\Gamma') \ , 
\end{align}
which are, respectively, the displacement vectors and covariance matrices of the states
\begin{align}
    (\psi_1,\psi_2,\psi_3)=(S_j(z)|\alpha\rangle, |\alpha'\rangle, S_j(z)|\psi\rangle) \ .
\end{align}
In Line~\ref{it:algSQu} we have 
\begin{align}
    u &= \overline{r} \\
    &= \langle  \psi, \alpha \rangle \\
    &= \langle S_j(z)\psi , S_j(z)\alpha\rangle \\
    &= \langle \psi_3, \psi_1 \rangle \ .
\end{align}
In Line~\ref{it:algSQv} we have 
\begin{align}
    \label{eq:squeezedcoherentstateCoherentStOverlap} 
    v &= 1/\sqrt{\cosh(z)} \\
    &= \langle S_j(z)0, 0\rangle \\
    &= \langle D(-\alpha')S_j(z)0, D(-\alpha')0\rangle \\
    &= \langle S_j(z) D(-\alpha) 0, D(-\alpha')0\rangle \\
    &=\langle  S_j(z) \alpha, \alpha' \rangle \\
    &=\langle  \psi_1, \psi_2 \rangle \ .
\end{align}
This is obtained by direct calculation considering~\cref{eq:squeezedcoherentstate} and 
\begin{align}
S (z) | 0 \rangle &= \frac{1}{\sqrt{\cosh{z}}} \sum_{k=0}^\infty (-\tanh(z))^k\frac{\sqrt{(2k)!}}{2^k k!} | 2k \rangle \ .
\end{align}
Then, by construction of the algorithm~$\Aoverlaptriple$ we have 
\begin{align}
    \langle \alpha, S_j(z)\psi \rangle = \Aoverlaptriple(\Gamma_1, \Gamma_2, \Gamma_3, d_1,d_2,d_3, u, v, 0)
\end{align}
which we set as~$r'$ in Line~\ref{it:algSQrp}. Hence, the description~$\Delta'$ satisfies Item~\eqref{it:DefDescCond3} in~\cref{def:description}.

The runtime of the algorithm is~$O(n^3)$ due to the use of the algorithm~$\Aoverlaptriple$.
\end{proof}

We give pseudocode for the algorithm~$\Aunitary$ and in~\cref{lem:Aunitary} we prove that it outputs the description of a Gaussian state evolved with a Gaussian unitary operation. We also give a diagram of the algorithm in~\cref{fig:evolve} when the operation is a displacement, phase shit or beamsplitter. When the operation is a single-mode squeezing operator, we refer to~\cref{fig:squeezing} for a diagram of the algorithm.

\begin{figure}[H]
\raggedright
\par\rule{\columnwidth}{1pt}
\textbf{Algorithm $\Aunitary$}
\vspace{-2mm}
\par\rule{\columnwidth}{1pt}
\textbf{Input: }{$\desc, \Delta_U, \ell$} \\ 
where
$\begin{cases}
    \desc=(\cov,\alpha,r) \in \Desc_n \\
    \text{$\Delta_U$ is a description of a Gaussian unitary} \\
    \ell\in\{\mathsf{displacement}, \mathsf{phaseshift},\allowbreak \mathsf{beamsplitter}, \mathsf{squeezing}\}
\end{cases}$ \\
\textbf{Output: }{a description of a pure Gaussian state}  
\begin{algorithmic}[1]
    \If{\label{alg:P-S0}$\ell = \mathsf{squeezing}$}
        \State{$(z,j)\leftarrow \Delta_U$}
        
        \Comment{$\Delta_U=(z,j)$ is a description of~$U=S_j(z)$}
        \State{\label{alg:P-S1}\textbf{return}~$\Asqueezing(\desc, z, j)$ }
    \EndIf
     \If{\label{alg:P-D1}$\ell = \mathsf{displacement}$}
        \State{$\beta \leftarrow \Delta_U$}
        \Comment{$\Delta_U=\beta$ is a description of~$U=D(\beta)$}
        \State{\label{alg:P-D2}$S\leftarrow I$}
        \State{\label{alg:P-D3}$\alpha'\leftarrow $ evaluate~\cref{eq:Dcoherentstate}}
    \EndIf
    \If{\label{alg:P-PS1}$\ell = \mathsf{phase shift}$}
        \State{$(\phi,j) \leftarrow \Delta_U$}

        \Comment{$\Delta_U = (\phi,j) $ is a description of~$U=F_j(\phi)$}
        \State{\label{alg:P-PS2}$S\leftarrow S$ given in the `phase shift' row of~\cref{tab:gaussian-unitaries}}
        \State{\label{alg:P-PS3}$\alpha'\leftarrow$ evaluate~\cref{eq:PScoherentstate}}
    \EndIf
    \If{\label{alg:P-BS1}$\ell = \mathsf{beamsplitter}$}
        \State{$(\omega,j,k) \leftarrow \Delta_U$}
        
        \Comment{$\Delta_U = (\omega,j,k) $ is a description of~$U=B_{j,k}(\omega)$}
        \State{\label{alg:P-BS2}$S\leftarrow S$ given in the `beamsplitter' row of~\cref{tab:gaussian-unitaries}}
        \State{\label{alg:P-BS3}$\alpha'\leftarrow$ evaluate~\cref{eq:BScoherentstate-alpha}}
    \EndIf
    \State{\label{alg:P-Gamap}$\cov' \leftarrow S \Gamma S^T$} 
    \Comment{covariance matrix of~$U|\psi\rangle$}
    \State{\textbf{return}~$(\cov', \alpha', r)$ }
\end{algorithmic}
\vspace{-3mm}
\par\rule{\columnwidth}{1pt}
%\caption{\justifying Given the description~$\Delta\in\Desc_n$ of a Gaussian state~$\psi\in\mathrm{Gauss}_n$, a description~$\Delta_U$ of a Gaussian unitary~$U$ and a label~$\ell\in\{\mathsf{displacement}, \mathsf{phaseshift},\allowbreak \mathsf{beamsplitter}, \mathsf{squeezing}\}$, the algorithm~$\Aunitary$ returns the description of the state~$U| \psi(\Delta) \rangle$.}
%\label{alg:Aunitary}
\end{figure}

\begin{theorem}
    \label{lem:Aunitary}
    The algorithm~$\Aunitary$ is such that
    \begin{align}
        \label{eq:prop-Aunitary}
        |\psi( \Aunitary(\desc, \Delta_U,\ell) )\rangle = U | \psi(\desc) \rangle
    \end{align}
    for all~$\desc \in \Desc_n$, $\Delta_U$ a description of a Gaussian unitary~$U$ and $\ell\in\{\mathsf{displacement}, \mathsf{phaseshift},\allowbreak \mathsf{beamsplitter}, \mathsf{squeezing}\}$. 
    That is, the output of the algorithm is a description of the state~$U | \psi(\Delta)\rangle $.
    When $\ell=\mathsf{squeezing}$ the algorithm has runtime $O(n^3)$, otherwise it has runtime~$O(n)$.
\end{theorem}

\begin{proof}
    The algorithm takes as input a description 
    \begin{align}
        \Delta = (\Gamma, \alpha, r) \in \Desc_n
    \end{align}
    of a Gaussian state~$|\psi\rangle$ and a description~$\Delta_U$ of a Gaussian unitary~$U$, and it outputs 
    \begin{align}
        \Delta' = (\Gamma', \alpha', r') \in \mathsf{Mat}_{2n\times 2n}(\bbR) \times \bbC^{n} \times \bbC \ . 
    \end{align}
    The algorithm differentiates between four cases: When~$\Delta_U$ is a description of: a single-mode squeezing operator (Lines~\ref{alg:P-S0} to~\ref{alg:P-S1}), a displacement (Lines~\ref{alg:P-D1} to~\ref{alg:P-D3}), a phase shifter (Lines~\ref{alg:P-PS1} to~\ref{alg:P-PS3}) and a beamsplitter (Lines~\ref{alg:P-BS1} to~\ref{alg:P-BS3}). 

    When~$U$ is a single-mode squeezing operator, the claim follows from the properties of the algorithm~$\Asqueezing$ (see~\cref{lem:Asqueezing}).

    Let us analyze the remaining cases. 
    In Lines~\ref{alg:P-D2}, \ref{alg:P-PS2} and~\ref{alg:P-BS2}, $S$ is the symplectic matrix associated to the unitary~$U$. Then, $\Gamma'$ computed in Line~\ref{alg:P-Gamap} is the covariance matrix of~$U|\psi(\Delta)\rangle$.
    The vector~$\alpha'\in\bbC^n$ computed in Lines~\ref{alg:P-D3}, \ref{alg:P-PS3} and~\ref{alg:P-BS3} labels the coherent state~$|\alpha'\rangle$ in 
    \begin{align}
        U | \alpha \rangle = | \alpha' \rangle \ .
    \end{align}
    Since~$|\alpha\rangle$ and~$|\psi(\Delta)\rangle$ have the same displacement (by the definition of~$\Delta$) and~$U$ maps states with the same displacement to states with the same displacement, then~$U|\psi(\Delta) \rangle$ has the same displacement~$\hat{d}(\alpha')$ as~$| \alpha' \rangle$. Then, $\Delta'$ satisfies Item~\eqref{it:DefDescCond1} in~\cref{def:description}.

    Finally, we set
    \begin{align}
        r' &= r \\
        &= \langle \alpha, \psi(\Delta) \rangle \\
        &= \langle U \alpha, U \psi(\Delta) \rangle \\
        &= \langle \alpha', U \psi(\Delta) \rangle
    \end{align}
    and observe this leads~$\Delta'$ to satisfy Item~\eqref{it:DefDescCond3} in~\cref{def:description}. Then, $\Delta'$ is the description of the state~$U|\psi(\Delta)\rangle$.

    The algorithm has runtime~$O(n)$ when $U$ is a displacement, a phase shift or a beamsplitter. This is because it involves computing sums and inner problems of vectors in~$\bbC^n$ / matrix multiplication with a matrix~$S$ that is non-trivial in a block of constant size (see Line~\ref{alg:P-Gamap}). 
\end{proof}

\begin{figure}[H]
\centering
\begin{subfigure}[t]{0.485\textwidth}
\centering
    \includegraphics[height=5cm]{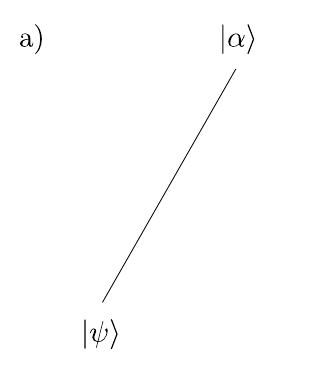}
    %\label{fig:evolve-a}
\end{subfigure}
\hfill
\begin{subfigure}[t]{0.485\textwidth}
\centering
    \includegraphics[height=5cm]{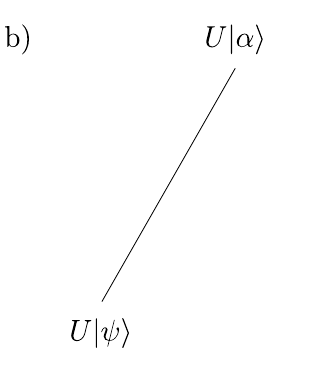}
    %\label{fig:evolve-b}
\end{subfigure}
\hfill
\caption{An illustration of the algorithm~$\Aunitary$ when~$U$ is a displacement, a phase shift or a beamsplitter operator. a) The input to the algorithm is the description of a Gaussian state~$|\psi\rangle$ which includes the overlap~$r = \langle \alpha, \psi\rangle$ with a coherent state~$|\alpha\rangle$, as well as a description~$\Delta_U$ of a Gaussian unitary~$U$. b) The reference state of the evolved state~$U|\psi\rangle$ is~$U|\alpha\rangle = |\alpha'\rangle$. The inner product~$r'=\langle U \alpha, U \psi \rangle$ is equal to~$r$ due to the unitarity of~$U$. The algorithm returns~$(\Gamma', \alpha', r)$, which is the description of the evolved state~$U|\psi\rangle$.}
\label{fig:evolve}
\end{figure}

\subsection{Computing descriptions of measured states}

Here we give an algorithm~$\Ameasure$ which, given the description of a Gaussian state~$|\psi\rangle$ and a vector~$\beta\in\bbC^k$ computes the description of the post-measurement state (see~\cref{sec:gaussian-measurements})
\begin{align}
    \label{eq:psi-post-measurement}
    |\psi'\rangle = \frac{1}{\sqrt{p(\beta)}} \Pi_\beta |\psi\rangle
\end{align}
after a heterodyne measurement of~$k$ modes. 

For convenience, we consider a subroutine~$\Aprob$ which takes as input a description~$\Delta\in\Desc_n$ of a Gaussian state and~$\beta\in\bbC^k$ and outputs the probability of obtaining measurement outcome~$\beta$ when performing a heterodyne measurement of the first~$k$ modes. More precisely,
\begin{align}
    \Aprob\left( \Delta, \beta \right) = \frac{1}{\pi^k} \| \Pi_\beta \psi(\Delta) \|^2  \ .
\end{align}
The algorithm simply evaluates~\cref{eq:heterodyne-prob} which takes time~$O(k^3)$.

We give pseudocode for the algorithm~$\Ameasure$, we prove that it computes the description of the post-measurement state in~\cref{lem:Ameasure} and we represent the algorithm pictorially in~\cref{fig:Ameasure}.

\begin{figure}[H]
\raggedright
\par\rule{\columnwidth}{1pt}
\textbf{Algorithm $\Ameasure$}
\vspace{-2mm}
\par\rule{\columnwidth}{1pt}
\textbf{Input: }{$\desc, \beta$} \\
where
$\begin{cases}
    \desc = \left(\cov = \begin{pmatrix} \Gamma_A & \Gamma_{AB} \\ \Gamma_{AB}^T & \Gamma_B \end{pmatrix}, \alpha=(\alpha_A,\alpha_B), r\right) \in \Desc_n \\
    \beta \in \bbC^k
\end{cases}$ \\
\textbf{Output: }{$(\cov', \alpha', r') \in \Desc_n$} 
\begin{algorithmic}[1]
    \State{$(s_A,s_B) \leftarrow (\hat{d}(\alpha_A),\hat{d}(\alpha_B))$}
    \State{\label{alg:meas-gamma}$\Gamma'\leftarrow$ evaluate~\cref{eq:heterodyne-cov}}
    \State{\label{alg:meas-s}$s'\leftarrow$ evaluate~\cref{eq:heterodyne-disp}}
    \State{\label{alg:meas-alphap}$\alpha'\leftarrow \hat{d}^{-1}(s')$}
    \State{\label{alg:meas-gamma123}$(\Gamma_1,\Gamma_2,\Gamma_3) \leftarrow (\Gamma,\Gamma',I)$}

    \Comment{covariance matrices of~$\left(|\psi(\Delta)\rangle, |\psi'\rangle, |\alpha'\rangle\right)$}
    \State{\label{alg:meas-d123}$(d_1, d_2, d_3), \leftarrow (s,s',s')$}
    
    \Comment{displacements of~$\left(|\psi(\Delta)\rangle, |\psi'\rangle, |\alpha'\rangle\right)$}

    \State{\label{alg:meas-gammavar}$\lambda\leftarrow \alpha-\alpha'$}
    \State{\label{alg:meas-u}$u\leftarrow \exp\left(i\im\left(\alpha'^T\overline{\alpha}\right)\right) r $}
    \State{\label{alg:meas-p}$p\leftarrow \Aprob(\Delta, \beta)$ }
    \State{\label{alg:meas-v}$v\leftarrow  \pi^k \sqrt{p}$}
    \State{\label{alg:meas-o}$\overline{r'}\leftarrow \Aoverlaptriple(\Gamma_1, \Gamma_2, \Gamma_3, d_1, d_2, d_3, u, v, \lambda)$}
    
    \Comment{compute~$\langle \alpha', \psi' \rangle$}
    \State{\textbf{return}~$(\cov', \alpha', r')$ }
\end{algorithmic}
\vspace{-3mm}
\par\rule{\columnwidth}{1pt}
%\label{alg:Ameasure}
\end{figure}

\begin{theorem}
\label{lem:Ameasure}
    The algorithm~$\Ameasure: \Desc_n \times \bbC^{k}$  runs in time~$O(n^3)$. The algorithm is such that
    \begin{align}
        \label{eq:propmeasure1}
        |\psi( \Ameasure(\desc, \beta) )\rangle = \frac{1}{\sqrt{p(\beta)}} \Pi_\beta |\psi(\Delta)\rangle
    \end{align}
    for all~$\desc \in \Desc_n$ and~$\beta\in\bbC^k$.
    That is, the output of the algorithm is a description of the post measurement state after a heterodyne measurement of the first~$k$ modes with outcome~$\beta$.
\end{theorem}

\begin{figure}[H]
\centering
\begin{subfigure}[t]{0.32\textwidth}
\centering
    \includegraphics[width=\textwidth]{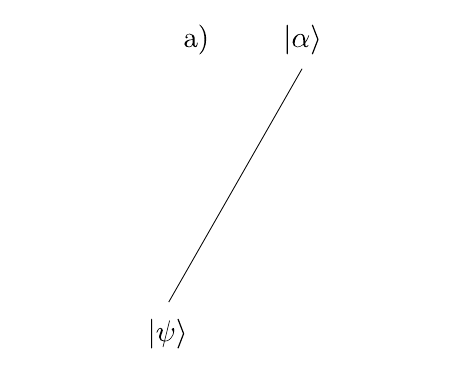}
    %\label{fig:measure-a}
\end{subfigure}
\hfill
\begin{subfigure}[t]{0.32\textwidth}
\centering
    \includegraphics[width=\textwidth]{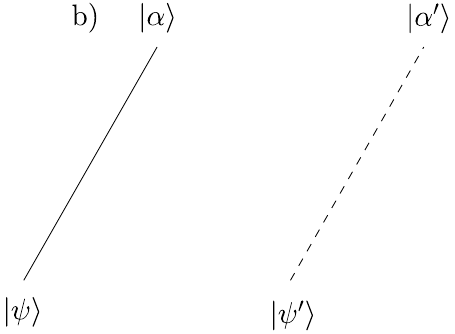}
    %\label{fig:measure-b}
\end{subfigure}
\hfill
\begin{subfigure}[t]{0.32\textwidth}
\centering
    \includegraphics[width=\textwidth]{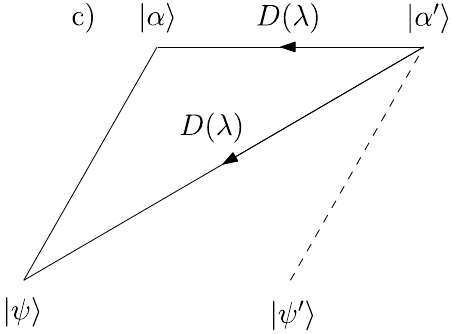}
    %\label{fig:measure-c}
\end{subfigure}
\hfill \\ \vspace{1mm}
\begin{subfigure}[t]{0.32\textwidth}
\centering
    \includegraphics[width=\textwidth]{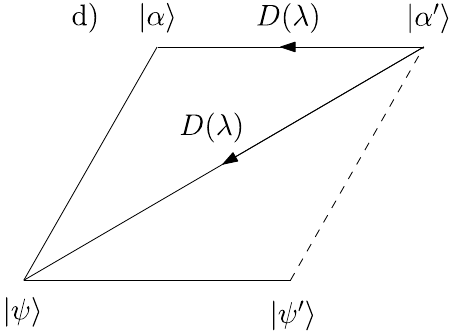}
    %\label{fig:measure-d}
\end{subfigure}
\hfill
\begin{subfigure}[t]{0.32\textwidth}
\centering
    \includegraphics[width=\textwidth]{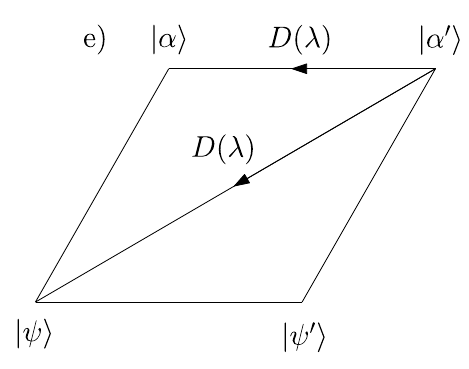}
    %\label{fig:measure-e}
\end{subfigure}
\begin{subfigure}[t]{0.32\textwidth}
\hfill
\end{subfigure}
\hfill
\caption{Pictorial description of the algorithm~$\Ameasure$. a) The algorithm takes as input the description~$\Delta$ of a Gaussian state~$\psi=\psi(\Delta)$ that includes the overlap~$r=\langle\alpha,\psi\rangle$ with the reference state~$|\alpha\rangle$. b) The covariance matrix~$\Gamma'$ and displacement of the post-measurement state~$|\psi'\rangle$ are computed using the covariance matrix formalism. The reference state~$|\alpha'\rangle$ is the coherent state with the displacement of~$|\psi'\rangle$. c) The algorithm computes~$\lambda\in \bbC^n$ and $\vartheta\in\bbR$ such that $D(\lambda) |\alpha\rangle = e^{i\vartheta}|\alpha'\rangle$. Then, $\langle \alpha', D(\lambda) \alpha \rangle = e^{i\vartheta}$ and $\langle \alpha', D(\lambda) \psi \rangle = e^{i\vartheta} r $ are known. d) The overlap $\langle \psi,\psi' \rangle = \|\psi'\| = \pi^k \sqrt{p}$ with $p$ the measurement probability is known. e) Finally, the overlap $r'=\langle \alpha', \psi' \rangle$ can be computed using the algorithm $\Aoverlaptriple$, which corresponds to completing the triangle. The algorithm returns $(\Gamma',\alpha',r')$.}
\label{fig:Ameasure}
\end{figure}

\begin{proof}
The algorithm takes as input~$(\Delta,\beta)\in \Desc \times \bbC^{k}$ and it outputs 
\begin{align}
    \Delta' = (\Gamma', \alpha', r') \in \mathsf{Mat}_{2n\times 2n}(\bbR) \times \bbC^{n} \times \bbC \ . 
\end{align}
We will show this is a description of the post-measurement state~$|\psi'\rangle$ given in~\cref{eq:psi-post-measurement}.
 
In Lines~\ref{alg:meas-gamma} and~\ref{alg:meas-s}  we set~$\Gamma'$ and~$s'$ to be respectively the covariance matrix and  displacement of the post-measurement state~$|\psi'\rangle$ (see~\cref{sec:gaussian-measurements}). In Line~\ref{alg:meas-alphap} the vector~$\alpha'$ labels the coherent state~$|\alpha'\rangle$ with displacement~$s'$. Then, the description~$\Delta'$ satisfies Item~\eqref{it:DefDescCond1} in~\cref{def:description}.

In Lines~\ref{alg:meas-gamma123} and~\ref{alg:meas-d123}~the matrices~$\Gamma_j,j\in[3]$ and vectors~$d_j,j\in[3]$ are, respectively, the covariance matrices and displacement vectors of the states
\begin{align}
   (\psi_1,\psi_2,\psi_3) = (|\psi(\Delta)\rangle, |\psi'\rangle, |\alpha'\rangle ) \ .
\end{align}
In Line~\ref{alg:meas-p}~we set~$p$ to be the probability density of obtaining measurement outcome~$\beta$. 
In Line~\ref{alg:meas-v} we have 
\begin{align}
    v &= \pi^k \sqrt{p} \\
    &= \frac{1}{\sqrt{p}} \| \Pi_\beta \psi(\Delta)\|^2  \\
    &= \langle \psi(\Delta), \psi' \rangle \\
    &= \langle \psi_1, \psi_2\rangle \ .
\end{align}
In Line~\ref{alg:meas-gammavar} we set~$\lambda = \alpha-\alpha'$ and in Line~\ref{alg:meas-u} we have
\begin{align}
    u &= \exp\left(i\im\left(\alpha'^T\overline{\alpha}\right)\right) r \\
    &= \exp\left(i\im\left(\alpha'^T\overline{\alpha}\right)\right) \langle \alpha, \psi(\Delta) \rangle \\
    &= \langle \alpha', D(\alpha-\alpha')  \psi(\Delta) \rangle \\
    &= \langle \psi_3, D(\lambda) \psi_1 \rangle \ ,
\end{align}
where we used the Weyl relations given in~\cref{eq:Wprop2b}.
Then, by the construction of the algorithm~$\Aoverlaptriple$ we have 
    \begin{align}
        \langle  \psi', \alpha' \rangle = \Aoverlaptriple(\Gamma_1, \Gamma_2, \Gamma_3, d_1,d_2,d_3, u, v, \lambda)
    \end{align}
    which we set as~$\overline{r'}$ in Line~\ref{alg:meas-v}. Hence, the description~$\Delta'$ satisfies Item~\eqref{it:DefDescCond3} in~\cref{def:description}.

    The runtime of the algorithm is~$O(n^3)$ due to the use of the algorithm~$\Aoverlaptriple$.
\end{proof}

\section{Superpositions of Gaussian states and measurements}
\label{sec:superpositions}
Here we establish various routines that can be used to deal with a superposition 
\begin{align}
    |\Psi\rangle = \sum_{j=1}^\chi c_j |\psi_j\rangle  \label{eq:psisuperpositiongauss}
\end{align}
of Gaussian states~$\psi_j\in\mathrm{Gauss}_n$. 
We assume that the vector~$\ket{\Psi}\in\cH_n$ is specified by the collection
\begin{align}
    \left\{ c_j, \Delta(\psi_j)\right\}_{j\in[\chi]} \ \label{eq:psivectorspecification}
\end{align}
of coefficients and descriptions of the individual terms in the superposition. We describe how to compute or estimate the probability density functions when applying a homodyne measurement to  a  superposition of Gaussian states specified by this data. We also explain how to update  descriptions of such superpositions after a Gaussian unitary or measurement is applied. This relies on the routines outlined in~\cref{sec:keeptrackphases} for tracking relative phases in linear optics, as well as (in the case of measurements) on the norm estimation algorithms laid out in~\cref{sec:normestimationalgorithmnaive,sec:fastnormestimationalg}.
For a more detailed explanation of this procedure, we refer readers to~\cite[Sections~4.1 and~4.3]{dias2023classical} where analogous 
subroutines are discussed in the context of fermionic linear optics.

In~\cref{sec:normestimationalgorithmnaive}, we sketch how to (exactly) compute the norm~$\|\Psi\|$ of~$\ket{\Psi}$ from the data~\eqref{eq:psivectorspecification}.  The corresponding algorithm has a runtime scaling of~$O(\chi^2)$, i.e., quadratically in the number of terms in the superposition. We then discuss a faster  probabilistic approximation algorithm in~\cref{sec:fastnormestimationalg}, whose runtime is linear in~$\chi$. In~\cref{sec:measurementsuperpos}, we discuss an algorithm for computing output probabilities when measuring a subset or all of the modes using a heterodyne measurement.

\subsection{Na\"ive norm estimation algorithm\label{sec:normestimationalgorithmnaive}}
It is straightforward to show the following:
\begin{lemma}
There is an algorithm~$\Aexactnorm$
which, on input
$\left\{ c_j, \Delta(\psi_j)\right\}_{j\in[\chi]}$, computes the norm~$\|\Psi\|$ of~$\Psi=\sum_{j=1}^\chi c_j\ket{\psi_j}$ in time~$O(\chi^2 n^3)$.
\end{lemma}    
\begin{proof}
Because 
\begin{align}
    \|\Psi\|^2 &=\sum_{j,k=1}^\chi \overline{c_k}c_j \langle \psi_k,\psi_j\rangle\ ,
\end{align}
the norm~$\|\Psi\|$ can be determined by computing~$O(\chi^2)$ overlaps of the form~$\langle \psi_j,\psi_k\rangle$. Using the subroutine~$\Aoverlap$, each of these computations requires a time of order~$O(n^3)$.
\end{proof}

\subsection{A fast norm estimation algorithm\label{sec:fastnormestimationalg}}
In pioneering work \cite{bravyiImprovedClassicalSimulation2016}, Bravyi and Gosset proposed an algorithm for estimating the norm of a superposition of stabilizer states that takes time linear in the number of terms~$\chi$ in the superposition. 
In Ref.~\cite{bravyiComplexityQuantumImpurity2017a} the same authors proposed a similar algorithm to estimate the norm of superpositions of fermionic Gaussian states that also takes time linear in the number of terms~$\chi$ in the superposition. 
These algorithms correspond to an improvement over the naive norm computation which takes time~$O(\chi^2)$.

Here we give an algorithm called~$\Afastnorm$ which estimates the norm of a superposition
\begin{align}
    |\Psi\rangle = \sum_{j=1}^\chi c_j |\psi_j\rangle
\end{align}
of~$\chi$ pure Gaussian states~$\psi_j\in\mathrm{Gauss}_n$ with~$c_j\in\bbC$.
Central to the construction of this algorithm is the following probabilistic process, which is specified by a vector~$\Psi\in\cH_n$, $R>0$ and~$L\in\mathbb{N}$:
\begin{enumerate}[(i)]
\item
Consider i.i.d. real-random variables~$X_1,\ldots,X_L$, where each~$X_j$ for~$j\in [L]$ is obtained as follows: We pick~$\alpha\in B_R(0)$ uniformly at random (according to the Lebesgue measure) and set
\begin{align}
    \label{eq:Xj}
    X_j&=R^2|\langle \alpha,\Psi\rangle|^2\ .
\end{align}
\item
Set
\begin{align}
    \label{eq:overlineX}
\overline{X}:=\frac{1}{L}\sum_{j=1}^L X_j\ .
\end{align}
\end{enumerate}
Then the following holds:
\begin{lemma}\label{lem:fastnormprobabilistic}
Suppose~$E>0$ and
\begin{align}
\langle \Psi',H\Psi'\rangle\leq E\quad\textrm{where}\quad\Psi'=\Psi/\|\Psi\|\ .
\end{align}
Let~$\varepsilon>0$ and~$p_f\in (0,1)$ be given. 
Set 
\begin{align}
    R:=\sqrt{\frac{E}{\varepsilon}}\quad\textrm{ and }\quad    L:= \left\lceil \frac{E}{4\pi p_f\varepsilon^3} \right\rceil\ .
\end{align}
Then 
\begin{align}
    \Pr\left[(1-\varepsilon)\|\Psi\|^2\leq  \overline{X}\leq (1+\varepsilon)\|\Psi\|^2\right]\geq 1-p_f\ .
\end{align}
\end{lemma}

\begin{proof}
We prove the claim for single-mode states. The generalization to~$n$-mode states is straightforward.

We write the Hamiltonian 
\begin{align}
H&= 2 (a^\dagger a + I)  = 2  a^\dagger a
\end{align}
such that
\begin{align}
\langle \Psi,H\Psi\rangle &=\frac{1}{\pi}\int_{
\mathbb{C}} 2 |\alpha|^2  |\langle \alpha,\Psi\rangle|^2 d\alpha\ ,
\end{align}
see~\cref{eq:numberstaterep} in~\cref{app:momentlimitedheterodynemeasurement}.
Define
\begin{align}
p(\alpha)&=\frac{1}{\pi} \frac{|\langle \alpha,\Psi\rangle|^2}{\|\Psi\|^2}\quad\textrm{for}\quad \alpha\in\mathbb{C}\ .
\end{align}
Then~$p$ is a probability density function (the one associated with the outcome distribution when applying the heterodyne measurement to the normalized state~$\Psi/\|\Psi\|$), and we have
\begin{align}
\frac{\langle\Psi,H\Psi\rangle}{2\|\Psi\|^2}&=\bbE_{\alpha\sim p}\left[|\alpha|^2\right]\ .
\end{align}
In particular, Markov's inequality implies that
\begin{align}
\Pr_{\alpha\sim p}\left[|\alpha|^2\geq R^2\right]&\leq \frac{\langle \Psi,H\Psi\rangle}{2 R^2\|\Psi\|^2}\ .\label{eq:markovalphareq}
\end{align}

Consider the random variable~$X=X_R$ defined by
choosing~$\alpha$ uniformly from~$B_R(0)$ and outputting
\begin{align}
X&= R^2 |\langle \alpha,\Psi\rangle|^2 \ .
\end{align}
The random variable~$X$ has expectation value
\begin{align}
\bbE\left[X\right]&=\frac{1}{\pi R^2}
\int_{B_R(0)} R^2 |\langle \alpha,\Psi\rangle|^2
d\alpha\\
&=\frac{1}{\pi }\int_{B_R(0)}  |\langle \alpha,\Psi\rangle|^2
d\alpha\\
&=\|\Psi\|^2\int_{B_R(0)} p(
\alpha)d\alpha\ ,
\end{align}
and it follows from~\cref{eq:markovalphareq} and the fact that~$p$ is a probability density function on~$\mathbb{C}$ that
\begin{align}
\|\Psi\|^2 \left(1-\frac{\langle\Psi,H\Psi\rangle}{2R^2
\|\Psi\|^2}\right)&\leq \bbE[X]\leq \|\Psi\|^2\ .\label{eq:multiplicativeerrorpsihpsi}
\end{align}
\noindent 
We have 
\begin{align}
\bbE[X^2]&=\frac{1}{\pi R^2}\int_{B_R(0)}
R^4 |\langle \alpha,\Psi\rangle|^4 d\alpha\\
&=\pi R^2 \|\Psi\|^4 \int_{B_R(0)}
p(\alpha)^2
d\alpha\\
&\leq \pi R^2\|\Psi\|^4
\end{align}
where we used
that~$\int_{B_R(0)} p(\alpha)^2 d\alpha \leq\int_{B_R(0)} p(\alpha) d\alpha \leq 1$.  Thus
\begin{align}
\var(X)&\leq \bbE[X^2]-\bbE[X]^2\\
&\leq \pi R^2 \|\Psi\|^4\ .
\end{align}
Let~$X_1,\ldots,X_L$ be i.i.d.~realizations of~$X$
and
\begin{align}
\overline{X}:=\frac{1}{L}\sum_{j=1}^L X_j \ .
\end{align}
Then
\begin{align}
\bbE\left[\overline{X}\right]&= \bbE\left[X\right]\\
\var(\overline{X})&=\var(X)/L
\leq \pi R^2 \|\Psi\|^4/L\ ,
\end{align}
and it follows by Chebyshev's inequality that
\begin{align}
\Pr\left[|\overline{X}-\bbE[\overline{X}]|\geq \kappa \|\Psi\|^2\right]&\leq \frac{\var(\overline{X})}{\kappa^2\|\Psi\|^4} =\frac{\pi R^2}{\kappa^2 L}\ .\label{eq:chebyscheffm}
\end{align}
Since
\begin{align}
\left|\overline{X}-\|\Psi\|^2\right|
&\leq \left|\overline{X}-\bbE[\overline{X}]\right|+
\left|\|\Psi\|^2-\bbE[\overline{X}]\right|\\
&= \left|\overline{X}-\bbE[\overline{X}]\right|+
\left|\|\Psi\|^2-\bbE[X]\right|\\
&\leq  \left|\overline{X}-\bbE[\overline{X}]\right|+
\|\Psi\|^2 \cdot \frac{\langle \Psi,H\Psi\rangle}{2 R^2\|\Psi\|^2}
\end{align}
by~\cref{eq:multiplicativeerrorpsihpsi}, it follows that
\begin{align}
\Pr\left[|\overline{X}-\|\Psi\|^2|
\geq \left(\varepsilon/2+
\frac{\langle\Psi,H\Psi\rangle}{2 R^2\|\Psi\|^2}
\right)\|\Psi\|^2
\right]&\leq \frac{4\pi R^2}{\varepsilon^2 L}\ ,\label{eq:usefullowerbound}
\end{align}
where we substituted~$\varepsilon/2$ for~$\kappa$ in~\cref{eq:chebyscheffm}.

Thus Eq.~\eqref{eq:usefullowerbound} can be rewritten using the definition~$\Psi'=\Psi/\|\Psi\|$ as
\begin{align}
\Pr\left[|\overline{X}-\|\Psi\|^2|
\geq \left(\varepsilon/2+
\frac{\langle \Psi',H\Psi'\rangle}{2 R^2}
\right)\|\Psi\|^2
\right]&\leq \frac{4\pi R^2}{\varepsilon^2 L}\ .\label{eq:lasteqxbarpsi}
\end{align}
Equation~\eqref{eq:lasteqxbarpsi} implies that 
\begin{align}
\Pr\left[(1-\varepsilon)\|\Psi\|^2\leq  \overline{X}\leq (1+\varepsilon)\|\Psi\|^2\right]&\geq 1-p_f
\end{align}
for any  pair~$(R,L)$ such that
\begin{align}
R&\geq \sqrt{\frac{\langle \Psi',H\Psi'\rangle}{\varepsilon}}
\quad\textrm{and}\quad L\geq \frac{R^2}{4\pi p_f\varepsilon^2}\ .\label{eq:RLcondition}
\end{align}
It is clear that if~$E\geq \langle \Psi',H\Psi'\rangle$ is an upper bound on the energy of the normalized state~$\Psi'$, then the choice
\begin{align}
R:=\sqrt{\frac{E}{\varepsilon}}\quad\textrm{and}\quad L:=\left\lceil
\frac{E}{4\pi p_f\varepsilon^3}\right\rceil
\end{align}
satisfies condition~\eqref{eq:RLcondition}.  This implies the claim.
\end{proof}

It is clear that this can be turned into a probabilistic algorithm $\Afastnorm_E$ for estimating the quantity~$\|\Psi\|^2$ up to a multiplicative error. 
The algorithm~$\Afastnorm_E$ takes as input a description of a~$n$-mode superposition~$\Psi$ of~$\chi$ Gaussian states and parameters~$\varepsilon>0, p_f\in(0,1)$, and it returns an estimate of the norm squared~$\|\Psi\|^2$ up to multiplicative error~$\varepsilon$ with success probability~$1-p_f$. It requires an upper bound~$E$ on the energy of the normalized state~$\Psi/\|\Psi\|$.
We establish the associated claim in~\cref{lem:fastnorm}. \Cref{lem:fastnorm} is an immediate consequence of~\cref{lem:fastnormprobabilistic}.

\begin{figure}[H]
\raggedright
\par\rule{\columnwidth}{1pt}
\textbf{Algorithm $\Afastnorm_E$}
\vspace{-2mm}
\par\rule{\columnwidth}{1pt}
\textbf{Input: }{$\{c_j, \Delta_j\}_{j=1}^\chi,\varepsilon, p_f$} \\
where
$\begin{cases}
    \{c_j, \Delta_j=(\Gamma_j,d_j,r_j)\}_{j=1}^\chi \in (\bbC \times \Desc_n)^\chi \\
    \varepsilon > 0, p_f\in(0,1), E>0
\end{cases}$ \\
\textbf{Output: }{$\overline{X} \in \bbC$} 
\begin{algorithmic}[1]
    \State{\label{alg:fastnorm-R}$R\leftarrow\sqrt{\frac{E}{\varepsilon}}$}
    \State{\label{alg:fastnorm-L}$L\leftarrow \big\lceil\frac{E}{4\pi p_f\varepsilon^3}\big\rceil$}
    \For{$\ell \leftarrow 1$ to~$L$}
        \State{\label{alg:fastnorm-alpha}$\alpha \leftarrow$ randomly sample from the Lesbegue measure on~$B_R(0) \subset \mathbb{C}^n$}
        \State{\label{alg:fastnorm-Dp}$ \Delta \leftarrow (I, d(\alpha), 1)$}
        \For{$j\leftarrow 1$ to~$\chi$}
            \State{\label{alg:fastnorm-oj}$o_j \leftarrow \Aoverlap(\Delta, \Delta_j)$}
        \EndFor
        \State{\label{alg:fastnorm-Xl}$X_l \leftarrow  R^2 \left| \sum_{j=1}^\chi c_j o_j \right|^2$}
    \EndFor
    \State{\label{alg:fastnorm-Xoverline}$\overline{X} \leftarrow \frac{1}{L}\sum_{\ell=1}^L X_\ell$}
    \State{\textbf{return}~$\overline{X}$ }
\end{algorithmic}
\vspace{-3mm}
\par\rule{\columnwidth}{1pt}
\label{alg:fastnorm}
\end{figure}

\begin{theorem}
\label{lem:fastnorm}
The algorithm~$\Afastnorm_E$ given 
 in~\cref{alg:fastnorm} achieves the following, given
as input a tuple~$\left\{ c_j, \Delta(\psi_j)\right\}_{j\in[\chi]}$ specifying a linear combination~$|\Psi\rangle = \sum_{j=1}^\chi c_j\ket{\psi_j}\in\cH_n$ of Gaussian states~$\psi_j,j\in[\chi]$ where  
\begin{align}
    \langle \Psi',H\Psi'\rangle \leq E\quad\textrm{with}\quad \Psi'=\Psi/\|\Psi\|\ ,\label{eq:energyupperboundpsihpsi}
\end{align}
and parameters~$\varepsilon>0$ and~$p_f\in(0,1)$,
the output~$\overline{X}$ of~$\Afastnorm_E$ satisfies
\begin{align}
    \label{eq:lem-alg-fastnorm-claim}
    (1-\varepsilon)\|\Psi\|^2\leq  \overline{X}\leq (1+\varepsilon)\|\Psi\|^2
\end{align}
with probability at least~$1-p_f$. The algorithm has runtime~$O\left(\frac{\chi n^3  E}{p_f \varepsilon^3}\right)$.
\end{theorem}

\textit{Note added:} The procedure in Threorem~\ref{lem:fastnorm} can be modified by using a median-of-means argument as done in \cite[Section 4.1]{bravyiSimulationQuantumCircuits2019a} (see also \cite[Lemma 10]{reardon-smithImprovedSimulationQuantum2023}). This gives an improved runtime of $O\left(\chi n^3  E \varepsilon^{-3} \log(1/p_f)\right)$.

\begin{proof}
Consider the superposition
\begin{align}
    | \Psi \rangle = \sum_{j=1}^\chi c_j | \psi(\Delta_j) \rangle 
\end{align}
and
\begin{align}
    \langle \Psi',H\Psi'\rangle \leq E\quad\textrm{where}\quad \Psi'=\Psi/\|\Psi\|\ .
\end{align}

In Line~\ref{alg:fastnorm-alpha} the algorithm samples~$\alpha$ from the Lebesgue measure on~$B_R(0)$.
In Line~\ref{alg:fastnorm-Dp} we set~$\Delta$ as the description of the coherent state~$|\alpha\rangle$. Then, by the properties of the algorithm~$\Aoverlap$ we have
\begin{align}
    o_j =  \langle \alpha,\psi_j\rangle
\end{align}
in Line~\ref{alg:fastnorm-oj}.
Then, in Line~\ref{eq:Xj} we have 
\begin{align}
    X_j&= R^2 |\langle \alpha,\Psi\rangle|^2 \\
    &= R^2 \left| \sum_{j=1}^n c_j \langle \alpha,\psi(\Delta_j)\rangle \right|^2 
\end{align}
which is an instance of the random variable given in~\cref{eq:Xj}.
Then, in Line~\ref{alg:fastnorm-Xoverline} we have~$\overline{X}$ given in~\cref{eq:overlineX}.
By choosing~$R$ and~$L$ respectively as in Lines~\ref{alg:fastnorm-R} and~\ref{alg:fastnorm-L}~\cref{lem:fastnormprobabilistic} applies, giving the claim. 

The algorithm has runtime 
\begin{align}
    L \cdot \chi \cdot \mathsf{time}(\Aoverlap) 
    &=
O\left(\frac{\chi n^3 E}{p_f \varepsilon^3}\right)\ ,
\end{align}
where we used~$\mathsf{time}(\Aoverlap) =O(n^3)$.
\end{proof}

We note that~$\Afastnorm_E$
implicitly appears to depend on the norm~$\|\Psi\|$
that it estimates as a consequence of the assumption~\eqref{eq:energyupperboundpsihpsi}, which itself involves the normalized state~$\Psi'=\Psi/\|\Psi\|$. 
In practice, it suffices to use a sufficiently large upper bound~$E$ on the energy of the normalized state~$\Psi'$. In~\cref{sec:measurementsuperpos}, we will discuss how this can be applied in the context of classically simulating Gaussian measurements on a superposition of Gaussian states.

Because of its central relevance to the complexity of our algorithm, let us briefly comment on how the energy (i.e., mean photon number) of the states occurring during the computation can be bounded. The corresponding bound depends on the mean photon number~$N$ of the initial superposition, as well as the unitaries applied.  Indeed, one can use that these unitaries correspond to energy-limited channels.

Recall we use the Hamiltonian given in~\cref{eq:hamiltonian}.
A channel~$\cN$ is said to be energy-limited if it maps energy bounded states to energy bounded states \cite{winter2017energyconstrained,shirokovExtensionQuantumChannels2020}. 
 Gaussian channels including finite displacements, phase shifters, beamsplitters and finite squeezing are energy-limited. Concretely, we have
\begin{alignat}{2}
    \tr(H D(\alpha)^\dagger \rho D(\alpha)) &= \tr(H\rho) + d(\alpha)^Td(\alpha) + 2 s^T d(\alpha) 
\end{alignat}
for all~$\alpha\in\bbC^n$,
where~$s\in\bbR^{2n}$ is the displacement of the state~$\rho$, and
\begin{alignat}{2}
    \tr(H P_j(\phi)^\dagger \rho P_j(\phi)) &= \tr(H \rho) 
    &&\quad\text{for all}\quad \phi\in\bbR \ , \\
    \tr(H B_{j,k}(\phi)^\dagger \rho B_{j,k}(\phi)) &= \tr(H \rho) 
    &&\quad\text{for all}\quad \omega\in\bbR \ , \\
    \tr(H S_j(z)^\dagger \rho S_{j}(z)) &\leq e^{2z} \tr(H \rho)
    &&\quad\text{for all}\quad z>1 \ ,
\end{alignat}
for all~$j,k\in[n]$.

Performing a heterodyne measurement and averaging over the measurement result amounts to the channel
\begin{align}
    \cE(\rho)&=\frac{1}{\pi}\int_{\mathbb{C}}
    \proj{\alpha} \rho \proj{\alpha } d\alpha\ .
    \label{eq:heterodynemeasurementchannel}
\end{align}
This channel is also energy-limited because we have 
\begin{align}
\tr(H\cE(\rho))&=\tr(H\rho)+2 \ . \label{eq:heterodynemeasurementenergy}
\end{align}
We refer to~\cref{app:momentlimitedheterodynemeasurement} for a proof of Eq.~\eqref{eq:heterodynemeasurementenergy}.

We note that somewhat more natural-looking expressions can be obtained for non-centered states by considering the variance
\begin{align}
    \var_\rho&:=\sum_{j=1}^{2n}  \var_\rho(R_j)= \sum_{j=1}^n \left(\tr(R_j^2\rho) - \tr(R_j\rho)^2 \right) \ .
\end{align}
Displacement, phase shift and beamsplitter channels leave~$\var_\rho$ invariant, that is
\begin{align}
\var_\rho&=\var_{\rho'}
\end{align}
with~$\rho'$ the state after applying one of these channels to~$\rho$. On the other hand, for the squeezing channel with squeezing parameter~$z>1$ we have
\begin{align}
\var_{\rho'}\leq e^{2z} \var_\rho \ .
\end{align}

\subsection{Measurement of a superposition\label{sec:measurementsuperpos}}
Here we describe an algorithm~$\Ameasureprob$ which, given 
a tuple~$\left\{ c_j, \Delta(\psi_j)\right\}_{j\in[\chi]}$ specifying a superposition~$|\Psi\rangle = \sum_{j=1}^\chi c_j |\psi_j\rangle$ of Gaussian states, computes output probabilities associated with measuring some (respectively all) modes using a heterodyne measurement.

In more detail, suppose we would like to estimate the probability density function when measuring~$k\leq n$~modes. Without loss of generality, let us consider the case where we measure the first~$k$ modes. The corresponding output distribution is absolutely continuous with respect to the Haar measure on~$\mathbb{C}^k$, with corresponding probability density function given by
\begin{align}
\label{eq:p-psi-alpha}
     p_\Psi(\alpha)&=\frac{1}{\pi^k}\|\Pi_\alpha\Psi\|^2\quad\textrm{where}\quad \Pi_\alpha=\proj{\alpha}
\otimes I^{\otimes n-k}\ 
\end{align}
for~$\alpha\in\mathbb{C}^k$. 

For any~$\alpha\in\mathbb{C}^k$, we have that 
\begin{align}
    \Pi_\alpha \Psi&=   \sum_{j=1}^\chi c_j \Pi_\alpha\psi_j=\sum_{j=1}^\chi c'_j(\alpha)\psi'_j(\alpha)
\end{align}
where
\begin{align}
    c'_j(\alpha) &= c_j\|\Pi_\alpha \psi_j\|\\    \psi'_j(\alpha)&=\Pi_\alpha\psi_j/\|\Pi_\alpha\psi_j\|
\end{align}
is a superposition of~$\chi$ (or fewer) Gaussian states~$\{\psi'_j(\alpha)\}_{j=1}^\chi$. Furthermore, descriptions of these states and the value of the coefficients~$\{c_j(\alpha)\}_{j=1}^\chi$
can be computed using the algorithm~$\Ameasure$.
Since the quantity~$p_\Psi(\alpha)$ in~\cref{eq:p-psi-alpha} essentially is the squared norm  of this superposition, this can then be estimated using a norm estimation routine.
We immediately obtain the following:
\begin{lemma}
Let~$E>0$.
There are algorithms~$\Ameasureprobexact$ and~$\Ameasureprobapproximate_E$
that, given as an input 
a tuple~$\left\{ c_j, \Delta(\psi_j)\right\}_{j\in[\chi]}$ specifying a state~$
 |\Psi\rangle=\sum_{j=1}^\chi c_j\ket{\psi_j}$ and~$\alpha\in\mathbb{C}$, satisfy the following.
 \begin{enumerate}[(i)]
\item The algorithm~$\Ameasureprobexact$ computes~$p_\Psi(\alpha)$
in time~$O(\chi^2 n^3)$.
\item 
Assume that the normalized post-measurement state
\begin{align}
\Psi'_\alpha&=\Pi_\alpha\Psi/\|\Pi_\alpha\Psi\|
    \end{align}
    satisfies 
    \begin{align}
    \langle \Psi'_\alpha,H\Psi'_\alpha\rangle &\leq E\ ,\label{eq:psiprimeHpsiprime}
        \end{align}
        the  output of~$\Ameasureprobapproximate_E$ is a value~$p$ that satisfies        ~$p/p_\Psi(\alpha)\in [1-\varepsilon,1+\varepsilon]$, except with probability~$p_f$.
 The runtime of the algorithm is~$O\left(\frac{\chi n^3 E}{p_f\varepsilon^3}\right)$. 
\end{enumerate}
\end{lemma}

\begin{proof}
    The algorithm~$\Ameasureprobexact$ corresponds to running the algorithm~$\Aexactnorm$ with input~$\{c_j'(\alpha), \Delta(\psi_j'(\alpha))\}_{j\in[\chi]}$ which computes~$\| \Pi_\alpha \Psi \|^2$ and outputting~$p_\Psi(\alpha) = \| \Pi_\alpha \Psi \|^2/\pi^k$.
    This requires determining each~$\|\Pi_\alpha \psi_j\|^2$ using the algorithm~$\Aprob$, i.e., 
    \begin{align}
        \|\Pi_\alpha \psi_j\|^2 = \pi^k \Aprob(\Delta(\psi_j(\alpha)), \alpha) \ ,
    \end{align}
    and determining the description of each post-measurement state~$\psi_j'(\alpha)$ using the algorithm~$\Ameasure$, i.e., 
    \begin{align}
        \Delta(\psi_j'(\alpha)) = \Ameasure(\Delta(\psi_j(\alpha)), \alpha) \ .
    \end{align}
    Then, the runtime of~$\Ameasureprobexact$ is
    \begin{align}
        &\chi \left(\mathsf{time} (\Ameasure) + \mathsf{time} (\Aprob)\right) 
        + \mathsf{time}(\Aexactnorm) = O(\chi^2 n^3) \ .
    \end{align}

    The algorithm~$\Ameasureprobapproximate_E$ corresponds to running the algorithm~$\Afastnorm_E$ with input~$\{c_j'(\alpha), \Delta(\psi_j'(\alpha))\}_{j\in[\chi]}$, $\varepsilon$ and~$p_f$. By~\cref{lem:fastnorm}, dividing by~$\pi^k$ gives a value~$p$ satisfying the claim.
    As for the exact algorithm, this requires~$\chi$ uses of the algorithms~$\Aprob$ and~$\Ameasure$.
    Then, the runtime of~$\Ameasureprobexact$ is
    \begin{align}
        &\chi \left(\mathsf{time} (\Ameasure) + \mathsf{time} (\Aprob)\right) 
        + \mathsf{time}(\Afastnorm_E) = O\left(\frac{\chi n^3 E}{p_f\varepsilon^3}\right) \ .
    \end{align}
\end{proof}

Application of the algorithm $\Ameasureprobapproximate_E$
requires an upper bound~$E$ on the energy~$\langle \Psi'_\alpha,H\Psi'_\alpha\rangle$ of the normalized post-measurement state~$\Psi'_\alpha$ associated with the outcome~$\alpha\in\mathbb{C}^k$, see Condition~\eqref{eq:psiprimeHpsiprime}.
We note that an upper bound on the energy~$\langle \Psi,H\Psi\rangle$ of the (final) state~$\Psi$ before the measurement can easily be obtained 
using the fact that the  involved operations are energy-limited, see~\cref{sec:fastnormestimationalg}. 
However, this does not suffice to obtain an upper bound on~$\langle \Psi'_\alpha,H\Psi'_\alpha\rangle$ in general, since this quantity also depends on the specific outcome~$\alpha\in\mathbb{C}^k$ considered. 
Indeed, as argued in~\cref{sec:appendixvariancepostmeasure}, there are states~$\Psi$  with small mean photon number (energy) and~$\alpha\in\mathbb{C}^k$ such that the corresponding normalized post-measurement state~$\Psi'_\alpha$ has (arbitrarily) large energy.

Let us argue that this is not a problem in practice. First, we have 
\begin{align}
    \Pr\left[\|\alpha\|^2\geq R^2\right]&\leq \frac{\langle \Psi,H\Psi\rangle}{2R^2}\label{eq:upperboundinequalityalphanormR}
\end{align}for any~$R>0$
by the same reasoning as in the proof of Eq.~\eqref{eq:markovalphareq}. That is, most of the probability mass of the distribution~$p_\Psi$ is contained within a ball of radius~$R$ dictated by the energy of~$\Psi$. This means that the computational problem of determining~$p_\Psi(\alpha)$ will typically only be of interest for~$\alpha\in\mathbb{C}^k$ with small norm.
Furthermore, for such ``typical'' measurement outcomes~$\alpha\in\mathbb{C}^k$, the 
post-measurement state~$\Psi'_\alpha$
has small norm. Indeed, since the heterodyne measurement channel~$\cE$ is energy-limited (see~\cref{eq:heterodynemeasurementchannel,eq:heterodynemeasurementenergy})
we have 
\begin{align}
    \int_{\mathbb{C}^k} p_\Psi(\alpha)
\langle \Psi'_\alpha, H \Psi'_\alpha\rangle d\alpha &=\langle \Psi,H\Psi\rangle+1\ .\label{eq:markovinequalityexpectation}
    \end{align}
By Markov's inequality, this implies that the distribution~$p_\Psi$ has support concentrated on~$\alpha\in B_R(0) \subset \mathbb{C}^k$ with the property that the associated post-measurement state~$\langle \Psi'_\alpha,H\Psi'_\alpha\rangle$ has bounded energy (compared to that of~$\Psi$). 

We can summarize these observations by the following lemma. It shows that 
the estimation of~$p_\Psi(\alpha)$ works efficiently whenever~$\alpha$ belongs to a ``typical" set~$\cT_\delta$. That is, we have:
\begin{lemma}\label{lem:typicalsubsetlemma}
    Suppose 
    \begin{align}
\langle \Psi,H\Psi\rangle &\leq E\ .
    \end{align}
    Then, for any~$\delta>0$, there is a set~$\cT_\delta\subset\mathbb{C}^k$ with the following properties:
    \begin{enumerate}[(i)]
\item 
The set~$\cT_\delta$ has probability mass at least~$1-\delta$ under the distribution~$p_\Psi$.
\item 
$\cT_\delta\subset B_{\sqrt{E/\delta}}(0)$, i.e., each~$\alpha\in\cT_\delta$ has squared norm bounded by~$\|\alpha\|^2\leq E/\delta$.
\item 
There is a classical algorithm that, for any~$\alpha\in\cT_\delta$, outputs an estimate~$p$ of~$p_\Psi(\alpha)$ which satisfies~$p/p_\Psi(\alpha)\in [1-\varepsilon,1+\varepsilon]$, and its runtime is bounded by~$O(\chi n^3 E/(\delta p_f\varepsilon^3))$.
        \end{enumerate}
 \end{lemma}
\begin{proof}
Define 
\begin{align}
    \tilde{E}&=\frac{2(E+1)}{\delta} \ , \\
    R&=\sqrt{E/\delta}\ .
\end{align}
Then 
\begin{align}
\Pr\left[\langle \Psi_\alpha',H\Psi_\alpha'\rangle\geq \tilde{E}\right]&\leq \delta/2
\end{align}
by Markov's inequality and~\eqref{eq:markovinequalityexpectation}.
Furthermore, we have 
\begin{align}
\Pr\left[|\alpha\|^2 \geq R^2\right] &\leq 
\delta/2\ 
    \end{align}
    by Eq.~\eqref{eq:upperboundinequalityalphanormR}.
Defining the set~$\cT_\delta$ as 
\begin{align}
    \cT_\delta &:=\left\{
\alpha\in B_R(0)\ |\ 
\langle \Psi_\alpha',H\Psi_\alpha'\rangle\leq \tilde{E}   \right\}\subset \mathbb{C}^k\ ,
\end{align}
we thus have
\begin{align}
    \Pr\left[\alpha\in\cT_\delta\right]\geq 1-\delta\ 
\end{align}
by the union bound. Furthermore, for any~$\alpha\in \cT_\delta$, we can run~$\Ameasureprobapproximate_{\tilde{E}}$ obtaining, with probability at least~$p_f$, an estimate of~$p_\Psi(\alpha)$ to within multiplicative error~$\varepsilon$. The runtime is
\begin{align}
    O\left(\frac{\chi n^3 \tilde{E}}{p_f\varepsilon^3}\right)&=
    O\left(\frac{\chi n^3 E}{\delta p_f\varepsilon^3} \right)\ .
\end{align}
This is the claim.
\end{proof}

\textit{Note added:}
After posting our work to the arXiv, we were made aware of concurrent independent work by Hahn et al.~\cite{hahn2024classical}, which also establishes strong simulation algorithms for non-Gaussian operations along similar lines.

\section*{Acknowledgements}
BD and RK gratefully acknowledge support by the European Research Council under grant agreement no.~101001976 (project EQUIPTNT).

\appendix 

\section{Proof of~\cref{lem:fidelitybound}}
\label{app:proof-r-lowerbound}

Assume that~$\rho = \rho(\Gamma,d)$ has covariance matrix~$\Gamma\in\mathsf{Mat}_{2n\times 2n}(\mathbb{R})$ and displacement~$d\in\mathbb{C}^n$.
Recall that every coherent state~$|\alpha\rangle, \alpha \in\mathbb{C}^n$ is a Gaussian state~$\proj{\alpha}=\rho(I,\hat{d}(\alpha))$ with covariance matrix~$I$ and displacement vector~$\hat{d}(\alpha)$ with
\begin{align}
    \begin{gathered}
    \begin{aligned}
        \hat{d}_{2j-1}(\alpha)&=\sqrt{2} \re(\alpha_j)   \\
        \hat{d}_{2j}(\alpha)&=\sqrt{2} \im(\alpha_j)
    \end{aligned}
    \end{gathered}
    \quad\text{for}\quad
    j \in [n] \ .
\end{align}
It immediately follows from~\cref{eq:trace2gaussians} that the overlap~$\langle \alpha,\rho\alpha\rangle$ is maximized when choosing~$\alpha$ such that~$\hat{d}(\alpha)=d$, i.e., 
\begin{align}
\alpha_j&=\frac{1}{\sqrt{2}} (d_{2j-1}+id_{2j})\quad\textrm{for}\quad j\in [n]\ ,
\end{align}
and that for this choice of~$\alpha$, we have 
\begin{align}
|\langle \alpha,\psi\rangle|^2 &=\frac{2^{n}}{\sqrt{\det(I+\cov)}}\ .
\end{align}

Recall (see~\cref{sec:gaussian-states}) that the covariance matrix~$\Gamma$ of a pure~$n$-mode Gaussian state has the form
\begin{align}
    \cov &=KZK^T \ ,
\end{align}
where~$K\in Sp(2n) \cup O(2n)$ and~$Z=\diag(z_1, \ldots, z_n, z_1^{-1}, \ldots, z_n^{-1})$. Assuming that~$z_j>1$ for every~$j\in [n]$, it follows that
\begin{align}
\det(I+\cov)&=\prod_{j=1}^n (1+z_j)(1+z_{j}^{-1})=\prod_{j=1}^n (2+z_j+z_j^{-1}) \ .
\end{align}
Recall~\cref{eq:energy} which together with 
\begin{align}
    \langle \psi,H\psi\rangle  &= \sum_{j=1}^{2n} \langle \psi, R_j^2 \psi \rangle + n\ , \\
    d_j &= \langle \psi , R_j \psi\rangle \quad\text{for}\quad j \in [2n] 
\end{align}
gives
\begin{align}
    \tr(\Gamma) 
    &= 2 \left( \langle \psi,H\psi\rangle - d^T d - n \right) \\
    &= 2 \sum_{j=1}^{2n} \left( \langle \psi, R_j^2 \psi \rangle - \langle \psi, R_j \psi \rangle^2 \right) \\
    &= 2 \sum_{j=1}^{2n} \var_{\psi}(R_j) \ ,
\end{align}
where~$\var_{\psi}(R_j) = \langle \psi, R_j^2 \psi\rangle - \langle \psi, R_j \psi\rangle^2$.
Setting
\begin{align}
\lambda_j &=2+z_j+z_j^{-1}
\end{align}
we have 
\begin{align}
\lambda_j 
    &\geq 4\quad\textrm{ for all } j\in [n] \ ,\\
\det(I+\Gamma)
    &=\prod_{j=1}^n \lambda_j \ ,\\
\sum_{j=1}^{n}\lambda_j 
    &= \tr(\Gamma) + 2n \\
    &= 2 \sum_{j=1}^{2n} \var_{\psi}(R_j) + 2n \\
    &=: \Lambda \ .
\end{align}
It follows that
\begin{align}
\det(I+\Gamma)&\leq \sup_{\substack{
\lambda_j\geq 4\textrm{ for all }j\in [n] \\
\sum_{j=1}^n \lambda_j=\Lambda
}} \prod_{j=1}^n\lambda_j\ .
\end{align}
Let~$j<k$ be fixed and consider the function
\begin{align}
g(\lambda_1,\ldots,\lambda_n)&:=\prod_{\ell=1}^n \lambda_\ell\\
&=h(\lambda_j,\lambda_k)\prod_{\ell\in \{1,\ldots,n\}\backslash \{j,k\}}\lambda_\ell
\end{align}
where~$h(x,y):=xy$.
It is straightforward to check that for any~$\varepsilon\geq 8$
\begin{align}
\sup_{x,y \geq 4, x+y=\varepsilon}h(x,y)&=h(\varepsilon-4,4) \ . 
\end{align}
It follows by induction that 
\begin{align}
&\sup_{\substack{\lambda_j\geq 4\\
\sum_{j=1}^n \lambda_j=\Lambda}}g(\lambda_1,\ldots,\lambda_n)
=g(\Lambda,4,\ldots,4) \\
&\qquad\qquad= (\Lambda-4(n-1)) 4^{n-1} \\
&\qquad\qquad= \left(2 \sum_{j=1}^{2n} \var_{\psi}(R_j) - 2n + 4\right) 4^{n-1}
\end{align}
and we conclude that 
\begin{align}
\det(I+\cov) &\leq \left(2 \sum_{j=1}^{n} \var_{\psi}(R_j) - 2n + 4\right) 4^{n-1} \ .
\end{align}
This gives the claim
\begin{align}
|\langle \alpha,\psi\rangle|^2 &\geq \frac{2}{\sqrt{2 \sum_{j=1}^{n} \var_{\psi}(R_j) - 2n + 4}}\\
&= \frac{1}{\sqrt{ \frac{1}{2} \sum_{j=1}^{n} \left( \var_{\psi}(Q_j) + \var_{\psi}(P_j) - 1\right)  + 1}}\ .
\end{align}

\section{Proof of~\cref{prop:trace3statesWWeylPsi}}
\label{app:proof-trace3statesWWeylPsi}

Consider symmetric matrices~$\cov_1, \cov_2, \cov_3 \in \msf{Mat}_{2n}(\mathbb{R})$ satisfying~$\cov_j+i\Omega \geq 0$ for~$j\in[3]$ and vectors~$d_1, d_2\in\bbR^{2n}$. 
In the following, we give an expression for
\begin{align}
    \tr\left( D(\alpha) \rho(\cov_1, d_1) \rho(\cov_2, d_2) \rho(\cov_3,0) \right)
    \quad\text{for}\quad
    \alpha \in \bbC^n \ .
\end{align}

By Parseval's relation in~\cref{eq:parseval}, we have 
\begin{align}
    \label{eq:traux1}
    &\tr\left( D(\alpha) \rho(\cov_1, d_1) \rho(\cov_2, d_2) \rho(\cov_3,0) \right) \\
    &= \frac{1}{(2\pi)^n} \int_{\mathbb{R}^{2n}} d^{2n} \xi \, \overline{\chi_{\rho(\cov_1, d_1)D(-\alpha)}( \xi )} \chi_{\rho(\cov_2, d_2) \rho(\cov_3,0)}( \xi ) \ .
\end{align}

The characteristic function of~$ D(\alpha) \rho(\cov_1, d_1)$ is 
\begin{align}
    &\chi_{\rho(\cov_1, d_1) D(-\alpha)} (\xi) \\
    &\qquad= \tr\left(  \rho(\cov_1, d_1) D(-\alpha) D(\hat{d}^{-1}(\xi))  \right) \\
    &\qquad= \exp(\hat{d}(\alpha)^T i\Omega \xi / 2 ) \tr\left( \rho(\cov_1, d_1) D(\hat{d}^{-1}(\xi) - \alpha) \right) \\
    &\qquad= \exp(\hat{d}(\alpha)^T i\Omega \xi / 2 ) \tr\left( \rho(\cov_1, d_1) D(\hat{d}^{-1}(\xi - \hat{d}(\alpha))) \right) \\
    &\qquad= \exp(\hat{d}(\alpha)^T i\Omega \xi / 2 ) \chi_{\rho(\cov_1, d_1)} (\xi-\hat{d}(\alpha)) \\
    \label{eq:characWeylOpGaussianState} &\qquad= \exp\left( -\frac{1}{4} (\xi-\hat{d}(\alpha))^T \Omega^T \cov \Omega (\xi-\hat{d}(\alpha)) - d^T i \Omega (\xi-\hat{d}(\alpha)) + \frac{1}{2}\hat{d}(\alpha)^T i\Omega \xi \right) \ ,
\end{align}

where in the first and fourth identities we used the definition of characteristic function given in~\cref{eq:chi}, in the second identity we used~\cref{eq:Wprop2} and in the last identity we used~\cref{eq:chiGaussianState}. 
The characteristic function of~$\rho(\cov_2, d_2) \rho(\cov_3,0)$ is given by~\cref{eq:chiaux2-v2} in~\cref{prop:chiRho1d1Rho2}. \Cref{prop:chiRho1d1Rho2} uses the following:

\begin{lemma}
    \label{lem:chiProp2}
    Consider states~$\rho_1, \rho_2 \in \mathcal{B}(\cH_n)$. The characteristic function~$\chi_{\rho_1 \rho_2}$ satisfies
    \begin{align}
        \label{eq:chiProp2} \chi_{\rho_1 \rho_2} (\xi) &= 
        \frac{1}{(2\pi)^n} \int_{\bbR^{2n}} d^{2n}\eta \, \chi_{\rho_1}(\eta) \chi_{\rho_2}(\xi-\eta) \exp\left(\frac{1}{2}\xi^T i\Omega \eta \right) \ .
    \end{align}
\end{lemma}

\begin{proof}
    From~\cref{eq:chi,eq:WeylTransform} we have 
    \begin{align}
        \label{eq:aux2}
        \chi_{\rho_1 \rho_2} (\xi) = \frac{1}{(2\pi)^{2n}} 
        \int_{\bbR^{2n}} d^{2n} \eta \, \chi_{\rho_1}(\eta) 
        \int_{\bbR^{2n}} d^{2n} \gamma \, \chi_{\rho_2}(\gamma) \tr\left( D(-\hat{d}^{-1}(\eta))  D(-\hat{d}^{-1}(\gamma)) D(\hat{d}^{-1}(\xi)) \right) \ .
    \end{align}
    From~\cref{eq:Wprop2} we have
    \begin{align}
        D(-\hat{d}^{-1}(\eta))D(-\hat{d}^{-1}(\gamma))D(\hat{d}^{-1}(\xi)) 
        &= \exp\left( -\frac{1}{2} \eta^T i\Omega \gamma + \frac{1}{2} (\eta+\gamma)^T i\Omega \xi  \right) D(\hat{d}^{-1}(\xi-\eta-\gamma)) \ .
    \end{align}
    We use this together with 
    \begin{align}
        \tr( D(\xi) ) &= \delta(\xi) 
        \quad\text{for any}\quad \xi\in\bbR^{2n} \ , \\
        \label{eq:delta}
        \frac{1}{2\pi} \int_{\bbR^{2n}} d^{2n} x \, f(x) \delta(x-a) &= f(a)
        \quad\text{for any}\quad a\in\bbR^{2n} \ ,
    \end{align}
    to write~\cref{eq:aux2} as
    \begin{align}
        \chi_{\rho_1 \rho_2} (\xi) 
        &= \frac{1}{(2\pi)^{2n}}  
        \int_{\bbR^{2n}} d^{2n} \eta \, \chi_{\rho_1}(\eta) 
        \chi_{\rho_2}(\xi-\eta) \exp\left( -\frac{1}{2} \eta^T i\Omega (\xi-\eta) - \frac{1}{2} (\eta+\xi-\eta)^T i\Omega \xi  \right)  \\
        &= \frac{1}{(2\pi)^{n}}  
        \int_{\bbR^{2n}} d^{2n} \eta \, \chi_{\rho_1}(\eta) 
        \chi_{\rho_2}(\xi-\eta) \exp\left( \frac{1}{2} \xi^T i\Omega \eta \right) \ ,
    \end{align}
    where in the second identity we used~$\xi^T \Omega \xi =0$ for any~$\xi\in\bbR^{2n}$. This gives the claim.
\end{proof}

\begin{lemma} 
\label{prop:chiRho1d1Rho2}
Consider symmetric matrices~$\cov_2, \cov_3 \in \msf{Mat}_{2n}(\mathbb{R})$ and vectors~$\xi, d_2\in \mathbb{R}^{2n}$. The following holds:
\begin{align}
    \label{eq:chiaux2-v2}  \chi_{\rho(\cov_2, d_2) \rho(\cov_3,0) } \left( \xi \right) 
    = \frac{ \exp\left( - \frac{1}{4}  \xi^T \Omega^T\cov_4 \Omega\xi - d_2^T (\cov_2+\cov_3)^{-1} d_2 + i d_2^T (\cov_2+\cov_3)^{-1} (\cov_3-i\Omega) \Omega \xi \right)  }{\sqrt{\det(\Omega^T(\cov_2+\cov_3)\Omega/2)}} 
\end{align}
 with 
\begin{align}
    \cov_4 = \cov_3 - (\cov_3+i\Omega)(\cov_2+\cov_3)^{-1} (\cov_3-i\Omega) \ .
\end{align}
\end{lemma}

\begin{proof} 
We use~\cref{lem:chiProp2} to write 
\begin{align}
    \chi_{\rho(\cov_2, d_2) \rho(\cov_3,0)} \left( \xi \right) 
    &= \frac{1}{(2\pi)^n} \int_{\mathbb{R}^{2n}} d^{2n}\eta \, \chi_{\rho(\cov_2, d_2)}( \eta ) \chi_{\rho(\cov_3,0)}( \xi-\eta ) \exp\left(\frac{1}{2}\xi^T i\Omega \eta  \right) \ .
\end{align}
Applying~\cref{eq:chiGaussianState} gives
\begin{align}
    &\chi_{\rho(\cov_2, d_2) \rho(\cov_3,0)} \left( \xi \right) \\
    &= \frac{1}{(2\pi)^n} \int_{\mathbb{R}^{2n}} d^{2n}\eta \, 
    \exp\left( -\tfrac{1}{4} \eta^T  \Omega^T \cov_2 \Omega  \eta - d_2^T i\Omega \eta  -\tfrac{1}{4} (\xi-\eta)^T  \Omega^T  \cov_3 \Omega (\xi-\eta)  + \tfrac{1}{2}\xi^T i\Omega \eta  \right) \\
    &= \frac{\exp\left( -\tfrac{1}{4} \xi^T \Omega^T   \cov_3 \Omega \xi \right)}{(2\pi)^n}  \int_{\mathbb{R}^{2n}} d^{2n}\eta \, 
    \exp\left( -\tfrac{1}{4} \eta^T \Omega^T (\cov_2+\cov_3) \Omega \eta +\tfrac{1}{2} \xi^T \Omega^T  \cov_3 \Omega  \eta - d_2^T i\Omega \eta + \tfrac{1}{2}\xi^T i\Omega \eta  \right) \\
    \label{eq:chiIntAux1} &= \frac{\exp\left( - \tfrac{1}{4} \xi^T \Omega^T  \cov_3 \Omega \xi \right) }{(2\pi)^n} \int_{\mathbb{R}^{2n}} d^{2n}\eta \, 
    \exp\left( -\tfrac{1}{4} \eta^T \Omega^T (\cov_2+\cov_3) \Omega \eta + \left( \tfrac{1}{2}\xi^T \Omega^T ( \cov_3 + i \Omega) \Omega - d_2^T i\Omega \right) \eta \right) \ .
\end{align}
The following formula for a Gaussian integral will be useful,
\begin{align}
    \label{eq:GaussianIntegral} \int_{\mathbb{R}^{m}} d^{m}x \, \exp\left( - \frac{1}{2} x^T M x + y^T x \right) = \sqrt{\frac{(2\pi)^{m}}{\det M}} \exp\left( \frac{1}{2} y^T M^{-1} y\right) \ ,
\end{align}
where~$M\in\msf{Mat}_{m}(\mathbb{R})$ is a symmetric positive definite matrix, $y\in\mathbb{C}^m$ and~$m\in\bbN$.
Considering this together with
\begin{align}
    \label{eq:M} M&=\Omega^T(\cov_2 + \cov_3)\Omega /2 \ , \\
    \label{eq:y} y&= \frac{1}{2} \Omega^T( \cov_3 - i \Omega) \Omega \xi - i \Omega d_2 \ ,
\end{align}
and~\cref{eq:chiIntAux1} gives
\begin{align}
    &\chi_{\rho(\cov_2, d_2) \rho(\cov_3,0)} \left( \xi \right) =\frac{1}{\sqrt{\det(\Omega^T(\cov_2+\cov_3)\Omega/2)}}
    \exp\left(  - \frac{1}{4} \xi^T \Omega^T \cov_3 \Omega \xi + \frac{1}{2} y^T M^{-1} y\right) \ . 
\end{align}
Replacing~$M$ and~$y$ respectively by~\cref{eq:M,eq:y} and additional algebraic manipulation gives the claim. 
\end{proof}

Inserting~\cref{eq:characWeylOpGaussianState,eq:chiaux2-v2} into~\cref{eq:traux1} gives
\begin{align}
    &\tr\left( D(\alpha) \rho(\cov_1, d_1) \rho(\cov_2, d_2) \rho(\cov_3, 0) \right) \\
    &=   \frac{1}{(2\pi)^n\sqrt{\det(\Omega^T(\cov_2+\cov_3)\Omega/2)}} \\ &\qquad\cdot\int_{\mathbb{R}^{2n}} d^{2n} \xi 
    \exp\left( -\tfrac{1}{4} (\xi-\hat{d}(\alpha))^T \Omega^T \cov_1 \Omega (\xi-\hat{d}(\alpha)) + d_1^T i \Omega (\xi-\hat{d}(\alpha)) - \tfrac{1}{2}\hat{d}(\alpha)^T i\Omega \xi  \right)   \\
    &\qquad\qquad\qquad \,\,\,\,\exp\left( - \tfrac{1}{4} \xi^T \Omega^T \cov_4 \Omega \xi 
    - {d_2}^T (\cov_2 + \cov_3)^{-1} d_2
    + i {d_2}^T (\cov_2 + \cov_3)^{-1} ( \cov_3 - i \Omega) \Omega \xi \right)   \\
    &=  \frac{ \exp\left(  - {d_2}^T (\cov_2+\cov_3)^{-1} d_2 -\tfrac{1}{4}\hat{d}(\alpha)^T  \Omega^T  \cov_1 \Omega \hat{d}(\alpha) - i {d_1}^T \Omega \hat{d}(\alpha) \right) }{(2\pi)^n \sqrt{\det(\Omega^T(\cov_2+\cov_3)\Omega/2)}} \int_{\mathbb{R}^{2n}} d^{2n} \xi \, \exp\left( -\tfrac{1}{2} \xi^T M \xi + y^T \xi \right)
\end{align}
where
\begin{align}
    \label{eq:M2-v2} M &= \Omega^T (\cov_1 + \cov_4) \Omega /2 \ , \\
    \label{eq:y2-v2}y &= \Omega^T \left( i  (\cov_3+i\Omega) (\cov_2+\cov_3)^{-1} d_2 - i d_1 - \frac{1}{2} (\cov_1 + i\Omega) \Omega \hat{d}(\alpha)  \right) \ .
\end{align}
Then, applying~\cref{eq:GaussianIntegral} gives
\begin{align}
    &\tr\left( D(\alpha) \rho(\cov_1, d_1) \rho(\cov_2, d_2) \rho(\cov_3, 0) \right) \\
    &\qquad= \frac{ \exp\left(   - {d_2}^T (\cov_2+\cov_3)^{-1} d_2 -\frac{1}{4}\hat{d}(\alpha)^T  \Omega^T  \cov_1 \Omega \hat{d}(\alpha) - i {d_1}^T \Omega \hat{d}(\alpha)  +   \frac{1}{2} y^T M^{-1} y\right) }{\sqrt{\det(\Omega^T(\cov_2+\cov_3)\Omega/2) \det(\Omega^T(\cov_1+\cov_4)\Omega/2)}} 
    \ .
\end{align}
Replacing~$M$ and~$y$ respectively by~\cref{eq:M2-v2,eq:y2-v2} and additional algebraic manipulation gives
\begin{align}
\label{eq:trace3statesWWeylPsiaux}
    &\tr\left( D(\alpha) \rho(\cov_1, d_1) \rho(\cov_2, d_2) \rho(\cov_3, 0) \right) = \\
    &\frac{\exp\left( - {d_1 }^T W_1 {d_1 } - {d_2 }^T W_2 {d_2 } + d_1^T W_3 d_2 - \hat{d}(\alpha)^T \Omega^T \cov_5 \Omega \hat{d}(\alpha) - i\hat{d}(\alpha)^T \Omega^T  W_4 d_1  - i\hat{d}(\alpha)^T \Omega^T W_5 d_2  \right)}{\sqrt{\det(\Omega^T(\cov_2 +\cov_3 )\Omega/2) \det(\Omega^T(\cov_1 +\cov_4 )\Omega/2)}}
\end{align}
with~$\Gamma_5, W_1,W_2,W_3,W_4,W_5$ defined in \cref{eq:Gamma5,eq:W1,eq:W2,eq:W3,eq:W4,eq:W5}.

We use
\begin{align}
    D(\beta) D(\gamma) = \exp\left( - \hat{d}(\beta) i \Omega \hat{d}(\gamma)  \right) D(\gamma) D(\beta)
\end{align}
to show
\begin{align}
    &\tr\left( D(\alpha) \rho(\cov_1, d_1) \rho(\cov_2, d_2) \rho(\cov_3, d_3) \right) \\
    &\qquad\qquad=
    \exp\left(-\hat{d}(\alpha)^T i \Omega d_3 \right) \tr\left( D(\alpha) \rho(\cov_1, d_1-d_3) \rho(\cov_2, d_2-d_3) \rho(\cov_3, 0) \right)
\end{align}
which together with~\cref{eq:trace3statesWWeylPsiaux} gives the claim in~\cref{prop:trace3statesWWeylPsi}.

\section{Energy-limit on heterodyne measurement\label{app:momentlimitedheterodynemeasurement}}

Here we show the identity
\begin{align}
\tr(H\cE(\rho))=\tr(H\rho)+2
\end{align}
for the 
channel~$\cE(\rho)=\frac{1}{\pi}\int_{\mathbb{C}}
    \proj{\alpha} \rho \proj{\alpha } d\alpha$, 
see Eq.~\eqref{eq:heterodynemeasurementenergy}.
Here and below we use the following consequence of
\begin{align}
\frac{1}{\pi}\int_{\mathbb{C}} \proj{\alpha} d\alpha&=I\ .
\end{align}
We have 
\begin{align}
a a^{\dagger} &=\frac{1}{\pi}\int_{\mathbb{C}} |\alpha|^2\proj{\alpha} d\alpha\ . \label{eq:numberstaterep}
\end{align}

\begin{proof}
Let~$\varphi,\psi$ be arbitrary states. Then 
\begin{align}
\left\langle\varphi, a a^{\dagger} \psi\right\rangle & =\left\langle a^{\dagger} \varphi, a^{\dagger} \psi\right\rangle \\
& =\frac{1}{\pi} \int_{\mathbb{C}} d \alpha\left\langle a^{\dagger} \varphi, \alpha\right\rangle\left\langle\alpha, a^{\dagger} \psi\right\rangle \\
& =\frac{1}{\pi} \int_{\mathbb{C}} d \alpha\left\langle\varphi, a \alpha\right\rangle\left\langle a \alpha, \psi\right\rangle \\
& =\frac{1}{\pi} \int_{\mathbb{C}} d \alpha\left|\alpha\right|^2\langle\varphi, \alpha\rangle\langle\alpha, \psi\rangle\ .
\end{align}
\end{proof}
Consider the channel
\begin{align}
\cE(\rho)&=\frac{1}{\pi}\int_{\mathbb{C}}
\proj{\alpha}\rho\proj{\alpha}
d\alpha\ .
\end{align}
Then 
\begin{align}
\cE^\dagger(a^\dagger a)&=\cE(a^\dagger a) \\
&=\frac{1}{\pi}\int_{\mathbb{C}}
\proj{\alpha}a^\dagger a\proj{\alpha}
d\alpha\\
&=\frac{1}{\pi}
\int_{\mathbb{C}} |\alpha|^2 \proj{\alpha}d\alpha\\
&=aa^\dagger\ 
\end{align}
by Eq.~\eqref{eq:numberstaterep} and
\begin{align}
   \cE(a a^\dagger) 
   &= \cE(a^\dagger a + I) \\
   &= a a^\dagger + I \ .
\end{align}
Thus, for~$H= 2(a^\dagger a + I) = 2 a a^\dagger$ we have obtain
\begin{align}
\cE(H) &= 2 \cE(a a^\dagger) \\
&= 2 (a a^\dagger + I) \\
&= H + 2I \ .
\end{align}
This gives the claim~\eqref{eq:heterodynemeasurementenergy}.
Thus~$\cE$ is energy-limited.

\section{On the variance of  post-measurement states\label{sec:appendixvariancepostmeasure} }
We need a bound on the variance
\begin{align}
&\Var_{\Pi_\alpha\Psi/\|\Pi_\alpha\Psi\|}\\
&\quad:=\sum_{j=1}^n\left(\Var_{\Pi_\alpha\Psi/\|\Pi_\alpha\Psi\|}(Q_j)+\Var_{\Pi_\alpha\Psi/\|\Pi_\alpha\Psi\|}(P_j)\right)
\end{align}
of the post-measurement state~$\Pi_\alpha\Psi/\|\Pi_\alpha\Psi\|$.
Here we argue that information on the first two moments of~$\Psi$ is insufficient to bound this.

Let~$r\in\mathbb{R}$, $p\in [0,1]$ and~$z\geq 1$. We consider the family of states
\begin{align}
\ket{\Psi_{p,r,z}}&=\sqrt{1-p}\ket{0}\ket{0}+i\sqrt{p}\ket{r}S(z)\ket{0}\ \label{eq:psiprz}
\end{align}
where~$\ket{0}$ is the vacuum state and~$\ket{r}$ is the coherent state.
Then we have the following:
\begin{lemma}
Let~$z>1$ be arbitrary. For~$p\in (0,1)$, define~$\Psi(p)=\Psi_{p,1/p,z}$.
Then 
\begin{enumerate}[(i)]
\item\label{it:itemfirstconvergenceexp}
The covariance matrix~$\Gamma(\Psi(p))$ and displacement~$d(\Psi(p))$ of~$\Psi(p)$ satisfy
\begin{align}
    \begin{gathered}
    \begin{aligned}
        \Gamma(\Psi(p))&\ \rightarrow\   I_{4\times 4}\\
        d(\Psi(p))&\ \rightarrow\   (0,0,0,0)
    \end{aligned}
    \end{gathered}
\end{align}
for~$p\rightarrow 0$.
\item\label{it:claimtwoprojectionnonconvergence}
The normalized post-measurement state~$\Psi'(p):=\Pi_{1/p}\Psi(p)/\|\Pi_{1/p}\Psi(p)\|$ for measurement outcome~$1/p\in\mathbb{C}$ satisfies
\begin{align}
    \begin{gathered}
    \begin{aligned}
        \Var_{\Psi'(p)}&\ \rightarrow\  1 + \cosh(2z)\\
        \langle \Psi'(p),H\Psi'(p)\rangle&\ \rightarrow\  \infty
    \end{aligned}
    \end{gathered}
\end{align}
for~$p\rightarrow 0$.
\end{enumerate}
\end{lemma}
Since~$z>1$ was arbitrary, this shows that the variance and energy (mean photon number) of the post-measurement state can be arbitrarily large even though the first two moments are bounded.

\begin{proof}
Consider the state~$\Psi_{p,r,z}$ defined in~\eqref{eq:psiprz}. For brevity, let us write
\begin{align}
\ket{\Psi}&=\sqrt{1-p}\ket{0}\ket{\Psi_0}+\sqrt{p}\ket{r}\ket{\Psi_1}\ .
\end{align}
Here~$\Psi_0$ and~$\Psi_1$ 
are defined by
\begin{align}
\ket{\Psi_0}&=\ket{0}\\
\ket{\Psi_1}&=i S(z)\ket{0}\ .
\end{align}
That is, they are  normalized Gaussian states with relative phases such that 
\begin{align}
\Re \langle \Psi_0,\Psi_1\rangle &=0\ .\label{eq:relativephasesgzerozgone}
\end{align}
Using~\cref{eq:relativephasesgzerozgone} and the fact that~$\langle 0,r\rangle=e^{-r^2/2}\in\mathbb{R}$,  it is straightforward to check that~$\Psi$ is normalized for any~$p\in [0,1]$, $r\in\mathbb{R}$ and~$z\geq 1$.

For~$R=(Q_1,P_1,Q_2,P_2)$ and~$j\in [4]$ we have 
\begin{align}
\langle \Psi, R_j\Psi\rangle &=(1-p) \left(\bra{0}\bra{\Psi_0}\right)R_j\left(\ket{0}\ket{\Psi_0}\right)\\
&\qquad+p\left(\bra{1}\bra{\Psi_1}\right)R_j\left(\ket{1}\ket{\Psi_1}\right) \\
&\qquad+2\sqrt{p(1-p)}\Re\left(\left(\bra{0}\bra{\Psi_0}\right)R_j\left(\ket{1}\ket{\Psi_1}\right)\right)\\
&\rightarrow (\bra{0}\bra{\Psi_0})R_j(\ket{0}\ket{\Psi_0})\quad\textrm{ for }\quad p\rightarrow 0\ .
\end{align}
We similarly have, for any~$j,k\in [4]$, that 
\begin{align}
\langle &\Psi,\left(R_jR_k+R_kR_j\right)\Psi\rangle\\
&\qquad\rightarrow (\bra{0}\bra{\Psi_0}) (R_jR_k+R_kR_j)(\ket{0}\ket{\Psi_0})
\end{align}
for~$r\rightarrow 0$.
Using the definition~$\ket{\Psi_0}=\ket{0}$
the claim~\eqref{it:itemfirstconvergenceexp} follows by considering the covariance matrix and displacement of the state~$\ket{0}\ket{0}$. 

We have
\begin{align}
\Pi_r\ket{\Psi}&=\ket{r}(\sqrt{1-p}e^{-r^2/2}\ket{\Psi_0}+\sqrt{p}\ket{\Psi_1})\ .
\end{align}
Again using~\eqref{eq:relativephasesgzerozgone}, we obtain
\begin{align}
\|\Pi_r\ket{\Psi}\|^2&=(1-p)e^{-r^2}+p\ .
\end{align}
This means that the normalized post-measurement state is
\begin{align}
\Psi':=\Pi_r\Psi/\|\Pi_r\Psi\|&=\ket{r}
\left(\alpha(p,r)\ket{0}+i\beta(p,r)S(z)\ket{0}\right) \label{eq:projectedprpsi}
\end{align}
where
\begin{align}
\alpha(p,r)&:=
\frac{\sqrt{1-p}e^{-r^2/2}}{\sqrt{(1-p)e^{-r^2}+p}} \ , \\
\beta(p,r)&:=\frac{\sqrt{p}}{\sqrt{(1-p)e^{-r^2}+p}}\ .
\end{align}
Using that both~$\alpha(p,r)$ and~$\beta(p,r)$ are real, we have as before
\begin{align}
&\langle \Psi',R_j\Psi'\rangle \\
&= \alpha(p,r)^2 \left(\bra{r}\bra{0}\right)R_j\left(\ket{r}\ket{0}\right)\\
&\qquad+\beta(p,r)^2\left(\bra{r}\bra{S(z)0}\right)R_j\left(\ket{r}\ket{S(z)0}\right)\ ,
\end{align}
and similarly for the second moments. Thus the covariance matrix displacement of~$\Psi'$ are convex combinations of~$\ket{r}\ket{0}$ and~$\ket{r}S(z)\ket{0}$, respectively, i.e., 
\begin{align}
\Gamma(\Psi')&=\alpha(p,r)^2 I_{4\times 4}+\beta(p,r)^2 I_{2\times 2}\oplus \mathsf{diag}(e^{-2z},e^{2z}) \ , \\
d(\Psi')&=\alpha(p,r)^2 (\sqrt{2}r,0,0,0)+\beta(p,r)^2(\sqrt{2}r,0,0,0)\ ,
\end{align}
where we used that~$\alpha(p,r)^2+\beta(p,r)^2=1$.
It follows that
\begin{align}
\Var_{\Psi'}&=\tr(\Gamma(\Psi')) / 2 \\
&=2\alpha(p,r)^2+\beta(p,r)^2 (1+\cosh(2z))\ .
\end{align}
We also have 
\begin{align}
\langle \Psi',H\Psi'\rangle &= 2\alpha(p,r)^2+\beta(p,r)^2 (1+\cosh(2z))+2r^2 + 2
\end{align}
because
\begin{align}
    &\langle \Psi',H\Psi'\rangle = \sum_{j=1}^2 \langle \Psi' , (Q_j^2+P_j^2 + I) \Psi' \rangle\\
    &= \sum_{j=1}^2 \Big(\langle \Psi' , ((Q_j-d_{Q_j}I)^2+(P_j-d_{P_j}I)^2 + I) , \Psi'\rangle \\
    &\qquad\qquad\qquad +d_{Q_j}^2+d_{P_j}^2 \Big)\\
    &=\Var_\Psi'+ 2 \|\Psi'\|^2+ \sum_{j=1}^2(d_{Q_j}^2+d_{P_j}^2) \ .
\end{align}
It is straightforward to check that~$\lim_{p\rightarrow 0}\beta(p,1/p)=1$, hence the claim follows. 
\end{proof}

\bibliography{sample}
\bibliographystyle{unsrturl}

\end{document}